\documentclass[sn-mathphys-ay]{sn-jnl}%

\usepackage{graphicx}%
\usepackage{multirow}%
\usepackage{amsmath,amssymb,amsfonts,amsthm}%
\usepackage{mathrsfs}%
\usepackage[title]{appendix}%
\usepackage{xcolor}%
\usepackage{textcomp}%
\usepackage{manyfoot}%
\usepackage{booktabs}%
\usepackage{algorithm}%
\usepackage{algorithmicx}%
\usepackage{algpseudocode}%
\usepackage{listings}%

\usepackage{tikz}
\usetikzlibrary{positioning}
\usepackage{color}
\usepackage[dvipsnames]{xcolor}
\newtheorem{theorem}{Theorem}[section]
\newtheorem{proposition}[theorem]{Proposition}
\newtheorem{lemma}[theorem]{Lemma} 
\newtheorem{corollary}[theorem]{Corollary}

\newtheorem{remark}{Remark}

\DeclareMathOperator*{\argmin}{arg\,min}

\newcommand{\curlyL}{\mathcal{L}}

\newcommand{\EE}{\mathbb{E}}

\newcommand{\PP}{\mathbb{P}}
\newcommand{\RR}{\mathbb{R}}
\newcommand{\TT}{\mathbb{T}}
\newcommand{\VV}{\mathbb{V}}

\newcommand{\curlyO}{\mathcal{O}}

\newcommand{\cov}{\text{cov}}

\begin{document}

\title[Observer-Based Source Localization in Tree Infection Networks via Laplace Transforms]{Observer-Based Source Localization in Tree Infection Networks via Laplace Transforms}

\author[1]{\fnm{Kesler} \sur{O'Connor}}{}

\author[1]{\fnm{Julia M.} \sur{Jess}}{}

\author[1]{\fnm{Devlin} \sur{Costello}}{}

\author*[1]{\fnm{Manuel E.} \sur{Lladser}}\email{manuel.lladser@colorado.edu}

\affil[1]{\orgdiv{Department of Applied Mathematics}, \orgname{University of Colorado}, \orgaddress{ \city{Boulder}, \state{CO} \postcode{80309-0526},  \country{The United States}}}

\abstract{We address the problem of localizing the source of infection in an undirected, tree-structured network under a susceptible–infected outbreak model. The infection propagates with independent random time increments (i.e., edge-delays) between neighboring nodes, while only the infection times of a subset of nodes can be observed. We show that a reduced set of observers may be sufficient, in the statistical sense, to localize the source and characterize its identifiability via the joint Laplace transform of the observers’ infection times. Using the explicit form of these transforms in terms of the edge-delay probability distributions, we propose scale-invariant least-squares estimators of the source. We evaluate their performance on synthetic trees and on a river network, demonstrating accurate localization under diverse edge-delay models. To conclude, we highlight overlooked technical challenges for observer-based source localization on networks with cycles, where standard spanning-tree reductions may be ill-posed.}


\maketitle

%%%%%%%%%%%%%%%%%%%%%%%%%%%%%%%
\section{Introduction}
\label{sec1}

Interest in analyzing and understanding large-scale infections has persisted for decades. Extensive research has explored how infections grow and evolve as they spread across networks~\cite{kermack1927contribution, anderson1992infectious, leinhardt2013social, wasserman1994social, mollison1977spatial, newman2002spread, eubank2004modelling}. In contrast, source localization has received significantly less attention, despite the fact that identifying an infection source quickly and accurately is crucial for containment and the prevention of issues such as disease outbreaks and the spread of misinformation or contaminants. 

In recent years, various observer-based solutions have been proposed for the source localization problem. These methods, developed in~\cite{PRL, paluch, TRBS}, use the infection times of a typically sparse subset of nodes in an infection network to try to identify the source. To the best of our knowledge, Pinto, Thiran, and Vetterli introduced the first observer-based approach in 2012~\cite{PRL}. They employ a maximum likelihood estimator (MLE) derived from the joint probability density function (p.d.f.) of the observers' infection times and show that the MLE is optimal when the infection propagates over a tree---i.e., an undirected, connected graph without cycles---with independent but not necessarily identically distributed Gaussian propagation delays along the edges. Due to the complex interdependencies among the paths connecting the source and the observers, this remains the only case where an analytic expression for the joint distribution of observers' infection times is known. Nevertheless, the time complexity of this approach is linear in the number of nodes in the tree and it can be reduced by excluding observers with relatively large infection times~\cite{paluch}. An alternative approach to source localization is based on least squares, minimizing over all non-observer nodes the sum of squared differences between the observed and expected infection times of the observers~\cite{TRBS}.

Most source localization methods assume that the network over which an infection propagates is a tree, such as a river network or a pipeline. In practice, however, most networks through which an infection propagates---whether representing physical social interactions, contacts in online platforms, or computer networks---contain cycles, allowing transmission along a usually exponentially large number of possible paths between a source and each node. Nonetheless, the tree structure is technically appealing, particularly in susceptible-infected (SI) models, where the infection propagates along a growing tree that eventually spans the whole network. Because of this, source localization methods typically assume that the infection propagates along a spanning tree of the network. The criteria for selecting this tree vary widely in the literature, ranging from simple breadth-first search trees~\cite{PRL}, to shortest path trees~\cite{paluch}, to convex linear combinations of Gromov matrices~\cite{ji2019properties}, among others.

\medskip

\noindent{\textbf{Paper organization.}} In the remainder of the Introduction, we introduce details and notation for the tree infection model addressed in this work. In Section~\ref{sec:Sufficient}, we show how to identify redundant observers, reducing the source estimation problem to tree networks in which the observers are leaves, except possibly for a single observer. In Section~\ref{sec:Identifiability}, we address the identifiability of the source in terms of the Laplace transform of the vector of observers’ infection times. We then use the explicit form of this transform in Section~\ref{sec:Localization} to propose two source estimators using a least squares approach based on the empirical Laplace transform of the observers’ infection times. Section~\ref{sec:performance} is devoted to test our methods both in synthetic networks and an existing river network under various practical models of edge-delays. In Section~\ref{sec:limitations}, we highlighting technical challenges that have been overlooked in the literature for source localization in general networks. Finally, Section~\ref{sec:conclusions} presents concluding remarks, and Section~\ref{sec:proofs} contains the technical proofs of some of our preceding results.

This work is partially based on results and ideas from the recent theses~\cite{OConnor22, Jess24}. 

The implementation of all methods discussed in this manuscript, and the synthetic and real networks used to support our findings can be found in the GitHub repository:~\cite{Cos25}.

%%%%%%%%%%%%%%%%%%%%%%%%%%%%%%%
\subsection{Infection Model}

We assume an infection propagates between neighboring nodes in a fixed tree with \textit{vertex set} $V$ and \textit{edge set} $E$. Edges are undirected. The tree $T=(V,E)$ is known, finite, and undirected. A \textit{leaf} in $T$ is a node with precisely one neighbor (i.e., a node of degree 1). We denote the \textit{leaf set} of $T$ as $L$.

For nodes $u,v \in V$, we use $[u,v]$ to represent, depending on the context, the set of edges or vertices on the shortest path connecting $u$ and $v$. This path is unique because $T$ is a tree.  

The infection is assumed to begin at time zero from a single unknown node. We model its spread using an SI model originating at $s\in V$---the unknown source. For each edge $e=\{u,v\}\in E$, the infection propagates from an already infected node $u$ to a susceptible neighbor $v$ after a non-negative, random amount of time (or delay) having a continuous probability distribution. We denote this delay as $\tau_e$.
The random variables $\tau_e$, with $e\in E$, are assumed independent and to have known distributions. Since the SI model does not allow recovery, the infection continues to spread until every node in $T$ becomes infected. 

For each $v\in V$, define  
\[\tau_v := \sum_{e\in[s,v]} \tau_e.\]  
In other words, $\tau_v$ is the time of infection of node $v$. 
More generally, if $A\subset V$ is non-empty, define $\tau_A:=(\tau_v)_{v\in A}$. (In particular, $\tau_v = \tau_{\{v\}}$, although we continue to use the former notation.) Thus, $\tau_A$ is the vector of infection times of each node in $A$.

In our setting, infection times are observable only for nodes in a set $\curlyO\subset V$, the \textit{set of observers}, which we assume to be a nonempty but proper subset of $V$ to rule out trivial cases. In what follows, we write $\tau$ to denote $\tau_\curlyO$.

The \textit{source localization problem} we address requires estimating $s$ from a single realization of $\tau$. This contrasts with other approaches that assume observers know the nodes from which they were infected.

%%%%%%%%%%%%%%%%%%%%%%%%%%%%%%%
\section{Sufficient Statistics for Source Localization}
\label{sec:Sufficient}

In this section, we argue that only a handful of observer's infection times are usually needed for estimating $s$, as the remaining ones provide only redundant information about the source. For this, consider the following equivalence relation between non-observer nodes in $T$: for $u,v \in V \setminus \curlyO$, define
\[u \equiv v \text{ if and only if } [u,v] \cap \curlyO = \emptyset.\]
In what follows, the equivalence class of a node $u \in V \setminus \curlyO$ is denoted by $[u]$, and the collection of all equivalence classes is denoted by $[\curlyO]$.

For each $r\in[\curlyO]$, the \textit{boundary} of $r$, denoted $\partial r$, is the set of observers that are neighbors of a node in $r$. Each observer can be a neighbor of at most one node in each equivalence class; otherwise, there would be a cycle in $T$. More generally, if $R\subset[\curlyO]$ is non-empty, define
\[\partial R := \cup_{r\in R} \,\partial r.\]  

For each $o\in \curlyO$ and $R\subset[\curlyO]$, define
\begin{equation}
\label{def:VoR}
V_{o;R}:=\big\{v\in V\text{ such that }[o,v]\cap r=\emptyset,\text{ for each }r\in R\big\}. 
\end{equation}
Additionally, let $T_{o;R}=(V_{o;R},E_{o;R})$ be the subtree of $T$ rooted at $o$ with vertex set $V_{o;R}$; in particular, $o\in V_{o;R}$. In words, $T_{o;R}$ is the subtree of $T$ consisting of nodes that descend from $o$ which have no ancestor in an equivalence class of $R$.

\medskip

\begin{remark}
When $R\subset[\curlyO]$ is such that $|R|=1$, say $R=\{r\}$, we just write $T_{o;r}=(V_{o;r},E_{o;r})$ instead of $T_{o;\{r\}}=(V_{o;\{r\}},E_{o;\{r\}})$. 
\end{remark}

\medskip

Upon a realization of $\tau$, we say that a class $r\in[\curlyO]$ is \textit{feasible} when, for all $o \in \partial r$, if $o_1, o_2 \in V_{o;r} \cap \curlyO$ are such that $o_1$ is an ancestor of $o_2$ in $T_{o;r}$ then $\tau_{o_1} \le \tau_{o_2}$. If so, $\tau_{o_1}<\tau_{o_2}$ unless $o_1=o_2$ because $\tau_e>0$ for all $e\in E$ almost surely. The following result clarifies our terminology.

\medskip

\begin{lemma}
With probability one, $s$ cannot belong to a non-feasible equivalence class.
\label{lem:feasible}
\end{lemma}
\begin{proof}
Suppose $r$ is a non-feasible equivalence class. In particular, there exists $o\in\partial r$ and (distinct) $o_1,o_2\in V_{o;r}\cap \curlyO$ such that $o_2$ descends from $o_1$ in $T_{o;r}$ but $\tau_{o_2}<\tau_{o_1}$. Suppose that $s\in r$. Then $o_2$ cannot be infected before $o_1$, i.e., $\tau_{o_2}>\tau_{o_1}$. Since this contradicts the assumption that $r$ is non-feasible, we must have $s\notin r$.
\end{proof}

The next result provides a simple characterization of the feasible equivalence classes. We call a set $R\subset[\curlyO]$ a \textit{star arrangement} when $\cap_{r\in R}\,\partial r\ne\emptyset$. Any $R$ with $|R|=1$ is trivially a star arrangement. On the other hand, if $|R|>1$ is a star arrangement then $\left|\cap_{r\in R}\,\partial r\right|=1$; otherwise, there would be a cycle in $T$. 

\medskip

\begin{theorem}
With probability one, a class $r\in[\curlyO]$ is feasible if and only if $\argmin_{o\in \curlyO} \tau_o \in \partial r$. In particular, almost surely, at least one feasible equivalence class exists, and the feasible classes form a star arrangement.
\label{thm:equiv}
\end{theorem}

\medskip

In parametric statistics, a \textit{statistic} (i.e., a function of the data that does not depend on any unknown quantity to evaluate) is called \textit{sufficient} to estimate an unknown parameter when the conditional distribution of the data, given the statistic value, does not depend on the parameter~\cite{Cor22}. The following result characterizes the observers' infection times that are statistically sufficient for estimating the source.

\medskip

\begin{theorem}
Let $R\subset[\curlyO]$ be a star arrangement of classes. If $s\in\cup_{r\in R}\,r$ then $\tau_o$, with $o\in\partial R$, is a sufficient statistic for $s$.
\label{thm:sufficiency}
\end{theorem}

\medskip

\begin{figure}[]
\centering
\begin{tikzpicture}[scale = 1, bluenode/.style={circle, draw, fill=blue!0, inner sep = 0pt, minimum size = 17pt},
rednode/.style={circle, draw, fill=red!50, inner sep = 0pt, minimum size = 17pt},
observernode/.style={circle, draw, fill=gray!50, inner sep = 0pt, minimum size = 17pt},
yellownode/.style={circle, draw, fill=yellow, inner sep = 0pt, minimum size = 17pt},
blacknode/.style={circle,draw, fill=Cerulean, inner sep = 0pt, minimum size = 17pt},
rrednode/.style={circle,draw, fill=LimeGreen, inner sep = 0pt, minimum size = 17pt},
yyellownode/.style={circle,draw, fill=yellow, inner sep = 0pt, minimum size = 17pt},
purplenode/.style={circle,draw, fill=purple!80, inner sep = 0pt, minimum size = 17pt}]
% nodes
%row 1.5
\node [observernode] (51) at (-5,1.5)  {8};
\node [blacknode] (52) at (-3,1.5)  {};
%row 1 
\node [blacknode] (41) at (-4,1)  {};
\node [observernode] (42) at (1,1)  {3};
\node [observernode] (43) at (4,1)  {1};
%row 0
\node [observernode] (31) at (-6,0)  {9};
\node [blacknode] (32) at (-5,0)  {};
\node [blacknode] (33) at (-4,0)  {};
\node [observernode] (34) at (-3,0)  {7};
\node [observernode] (35) at (-2,0)  {6};
\node [observernode] (36) at (-1,0)  {5};
\node [bluenode] (37) at (0,0)  {};
\node [bluenode] (38) at (1,0)  {};
\node [observernode] (39) at (2,0)  {2};
\node [rrednode] (310) at (3,0)  {};
\node [rrednode] (311) at (4,0)  {};
\node [rrednode] (312) at (5,0)  {};
%row -1
\node [blacknode] (21) at (-5,-1)  {};
\node [blacknode] (22) at (-4,-1)  {};
\node [observernode] (23) at (0,-1)  {4};
\node [yyellownode] (24) at (2,-1) {};
\node [rrednode] (25) at (4,-1)  {};
%row -1.5
\node [yyellownode] (11) at (1,-1.5) {};
\node [yyellownode] (12) at (3,-1.5) {};
%row 1.5
\draw (51) -- (41) node {};
\draw (52) -- (41) node {};
%row 1
\draw (41) -- (33) node {};
\draw (42) -- (38) node {};
\draw (43) -- (311) node {};
%row 0
\draw (31) -- (32) node {};
\draw (32) -- (33) node {};
\draw (33) -- (34) node {};
\draw (34) -- (35) node {};
\draw (35) -- (36) node {};
\draw (36) -- (37) node {};
\draw (37) -- (38) node {};
\draw (38) -- (39) node {};
\draw (39) -- (310) node {};
\draw (310) -- (311) node {};
\draw (311) -- (312) node {};
%row -1
\draw (21) -- (32) node {};
\draw (22) -- (33) node {};
\draw (23) -- (37) node {};
\draw (24) -- (39) node {};
\draw (25) -- (311) node {};
%row -1.5
\draw (11) -- (24) node {};
\draw (12) -- (24) node {};
\end{tikzpicture}
\caption{Diagram of an infection tree with observer nodes labeled 1 through 9. It contains four equivalence classes with nodes colored blue, white, yellow, and green. The boundaries of these classes are $\{7,8,9\}$, $\{2,3,5,4\}$, $\{2\}$, and $\{1,2\}$, respectively. The white, yellow, and green classes form a star arrangement (centered at node 2). These three classes are feasible only when observer 2 is the first to become infected; in which case, observers 1 through 5 are sufficient to estimate the source. However, if observer 3 is the first to be infected, only the white class remains feasible, and observers 2 through 5 are sufficient to estimate the source.}
\label{fig:equiv}
\end{figure}

Theorems~\ref{thm:equiv}-\ref{thm:sufficiency} help reduce the complexity of the source localization problem in trees to only consider the infection times of observers in the boundary of (the star arrangement formed by) the equivalence classes that contain the first infected observer. Namely, the $\tau_o$, with $o\in\partial R$, where
\[R:=\bigcup_{r\in[\curlyO]:\,\argmin\limits_{w\in\curlyO}\tau_w\,\in\,\partial r}r.\]
In particular, the general source localization problem on trees is reduced to cases where the observers are all leaves---except possibly for a single interior node (the center of a star arrangement). However, the estimation problem may be substantially more difficult in the latter case. Indeed, suppose $|R|>1$ and let $o\in\partial R$ be the center of the arrangement. Then $\tau$ must satisfy 
\begin{equation}
\tau_o<\tau_w,\text{ for each }w\in \curlyO\setminus\{o\};
\label{ine:tauow}
\end{equation}
which makes its (conditional) distribution rather intractable. To overcome this issue one may be tempted to disregard the the infection time of the center of the arrangement, however, the statistic $\tau_{\curlyO\setminus\{o\}}$ is typically not sufficient for estimating the source when one conditions on (\ref{ine:tauow}).

To clarify the latter statement, consider the network in Figure~\ref{fig:stararrangement}. For simplicity, suppose that all edge-delays are independent and identically distributed (i.i.d.) with p.d.f. $f$. Let $f_{(i)|k}$ denote the p.d.f. of the $i$-th order statistic of $k$ i.i.d. random variables with p.d.f. $f$, and $\ast$ denotes the convolution operator between p.d.f.'s. (Recall that the convolution of multiple p.d.f.'s corresponds to the p.d.f. of the sum of independent random variables with distributions given by those p.d.f.'s.)

Then, conditioned on having 
\[\tau_0<\min_{1\le i\le n+1}\tau_i,\]
the distribution of $\tau_{\{1,\ldots,n+1\}}$ depends on the identity of the source, i.e., statistical sufficiency is lost when $\tau_0$ is disregarded. In fact, just focusing on the conditional distribution of $\tau_{n+1}$, one finds that 
\begin{align}    \PP\left(\tau_{n+1}=t\left|s=\ell,\,\tau_0<\min_{1\le i\le n+1}\tau_i\right.\right) 
\label{ide:tn+1l} &=2f(t)-f_{(2)|2}(t); \\
\PP\left(\tau_{n+1}=t\left|s=r,\,\tau_0<\min_{1\le i\le n+1}\tau_i\right.\right)
\label{ide:tn+1r} &=\big(f\ast f\ast f_{(1)|(n+1)}\big)(t).
\end{align}
There is, however, no reason for the p.d.f.'s in (\ref{ide:tn+1l}) and (\ref{ide:tn+1r}) to be equal. In fact, the only possible densities that could make these equal would have to be fixed points of the operator 
\[f\longrightarrow\frac{f_{(2)|2}}{2}+\frac{f\ast f\ast f_{(1)|(n+1)}}{2},\]
over the class of probability density functions supported on $[0, +\infty)$. This operator has no fixed points, however, because it does not preserve expected values---in fact, it increases them. Consequently, $\tau_{\{1,\ldots,n+1\}}$ is not sufficient for estimating the source in Figure~\ref{fig:stararrangement} when node $0$ is the first to get infected. 

\begin{figure}[t]
\centering
\begin{tikzpicture}[scale = 1, bluenode/.style={circle, draw, fill=blue!0, inner sep = 0pt, minimum size = 17pt},
observernode/.style={circle,draw, fill=gray!50, inner sep = 0pt, minimum size = 17pt},
yellownode/.style={circle,draw, fill=yellow, inner sep = 0pt, minimum size = 17pt},
greennode/.style={circle,draw, fill=LimeGreen, inner sep = 0pt, minimum size = 17pt}]
% Define positions of the vertices
\node[observernode] (6) at (-3, 0) {$\scriptstyle{n+1}$};
\node[yellownode] (L) at (-2, 0) {$\ell$};
\node[observernode] (0) at (-1, 0) {0};
\node[greennode] (R) at (0, 0) {$r$};
\node[observernode] (1) at (0, 1) {1};
\node[observernode] (2) at (0.71, 0.71) {2};
\node[observernode] (3) at (1, 0) {3};
\node[observernode] (5) at (0, -1) {$n$};
% Draw the edges
\draw (6) -- (L);
\draw (L) -- (0);
\draw (0) -- (R);
\draw (R) -- (1);
\draw (R) -- (2);
\draw (R) -- (3);
\draw (R) -- (5);
\draw [line width=2pt, line cap=round, dash pattern=on 0pt off 12pt] (1,-0.47) to[out=-90,in=0] (5);
\end{tikzpicture}
\caption{Toy diagram illustrating an infection tree where all the observer nodes, labeled $0$ through $(n+1)$, are leaves except for node $0$, which is the center of the star arrangement of equivalence classes when $\tau_0<\tau_i$ for $i=1,\ldots,(n+1)$.}
\label{fig:stararrangement}
\end{figure}
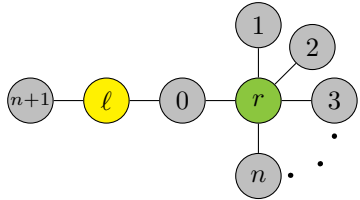

%%%%%%%%%%%%%%%%%%%%%%%%%%%%%%%
\section{Source Identifiability in Trees}
\label{sec:Identifiability}

In the context of statistical inference, the source is said to be \textit{identifiable} when the distribution of $\tau$ given that $s = v$ is unique for each $v \in V$. Unfortunately, unless the edge-delays are Gaussian~\cite{PRL}, explicitly computing the distribution of the vector $\tau$ is non-trivial, especially when the paths connecting observers to an alleged source overlap. Because of this, we use Laplace transforms to characterize the distribution of $\tau$ under each possible source. Importantly,  Laplace transforms uniquely determine the distribution of a random vector when they are finite in an open neighborhood of the origin. 

We emphasize that an analogous result can be formulated using characteristic functions~\cite{OConnor22}; however, we choose Laplace transforms because of the non-negative nature of infection times.

The \textit{Laplace transform} of $\tau$ is the function defined as
\begin{equation}
\label{def:actualLaplace}
\varphi(t):=\EE\big(e^{-\langle t,\tau\rangle}\big),\text{ for }t=(t_o)_{o\in\curlyO}\ge0;
\end{equation}
where $t\ge0$ means that $t_o\ge0$ for each $o\in\curlyO$. Since the source is unknown in our setting, we denote the above function as $\varphi_v(t)$ when assuming that $s = v$. Namely, for $v\in V\setminus\curlyO$:  
\[\varphi_v(t) := \EE\big(e^{-\langle t,\tau\rangle}\big| s = v\big), \text{ for }t\geq 0.\]

Our next result provides an explicit formula for the Laplace transform of $\tau$ under each possible source in terms of $\varphi_e$, for $e\in E$, i.e., the Laplace transform of the edge-delays along $T$. To state the result and implement our methods in software, it is convenient to introduce the following matrix with rows indexed by $\curlyO$ and columns indexed by $E$:
\[A_v(o,e) := \begin{cases}
    1, & \text{if edge } e \in [v,o];\\
    0, & \text{otherwise.}
\end{cases}\]

\medskip

\begin{theorem}
\label{thm:chf0conditioning}
For each $v\in V\setminus \curlyO$:
\begin{equation}
\label{eq:chfMLE1}
\varphi_v(t)
= \prod_{e\in E}\varphi_e\!\!\left(\sum_{o\in \curlyO(e|v)}t_o\right),\text{ for }t\ge0;
\end{equation}
where $\curlyO(e|v):=\big\{o\in\curlyO\text{ such that }A_v(o,e)=1\big\}$. In particular,
\begin{equation}
\label{eq:chfMLE2}
\sum_{o\in \curlyO(e|v)}t_o
= (tA_v)(e).
\end{equation}
\label{thm:chfMLE}
\end{theorem}

\begin{remark}
The elements in the set $\curlyO(e|v)$ are the observers that descend from $e$ when $T$ is rooted at $v$.
\end{remark}

\begin{proof}
If $s=v$ then
\[\langle t,\tau\rangle
=\sum_{o\in\curlyO}t_o\tau_o
=\sum_{o\in \curlyO}t_o\sum_{e\in E}A_v(o,e)\,\tau_e
=\sum_{e\in E}\tau_e\sum_{o\in \curlyO}t_o\,A_v(o,e)
=\sum_{e\in E} (tA_v)(e)\,\tau_e.\]
In particular, since $\tau_e$, with $e\in E$, are independent:
\[\varphi_v(t)
=\prod_{e\in E}\EE\left(e^{-(tA_v)(e)\,\tau_e}\right)
=\prod_{e\in E}\varphi_e\Big((tA_v)(e)\Big).\]
Since $(tA_v)(e)=\sum_{o\in \curlyO(e|v)}t_o$, the result follows. 
\end{proof}

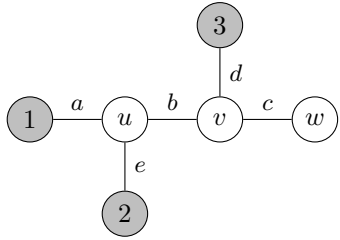
\begin{figure}[t]
\centering
\begin{tikzpicture}[scale = 1, 
observernode/.style={circle,draw, fill=gray!50, inner sep = 0pt, minimum size = 17pt}, whitenode/.style={circle,draw, fill=white, inner sep = 0pt, minimum size = 17pt}]
% nodes
\node [observernode] (1) at (-2.5,0)  {1};
\node [observernode] (2) at (-1.25,-1.25)  {2};
\node [observernode] (3) at (0,1.25)  {3};
\node [whitenode] (4) at (-1.25,0)  {$u$};
\node [whitenode] (5) at (0,0)  {$v$};
\node [whitenode] (6) at (1.25,0)  {$w$};
% edges
\path [-] (1) edge   node [above,pos=.5]  {{\small $a$}} (4);
\path [-] (4) edge   node [above,pos=.5]  {{\small $b$}} (5);
\path [-] (5) edge   node [above,pos=.5]  {{\small $c$}} (6);
\path [-] (4) edge   node [right,pos=.5]  {{\small $e$}} (2);
\path [-] (5) edge   node [right,pos=.5]  {{\small $d$}} (3);
\end{tikzpicture}
\caption{Example of an infection tree with observers labeled 1 to 3 (colored gray), non-observer nodes labeled $u$, $v$, and $w$, and edges set $\{a,b,c,d,e\}$. }
\label{fig:chartheoexample}
\end{figure}

To fix ideas about our last result, consider the infection tree in Figure~\ref{fig:chartheoexample}. Due to equation~(\ref{eq:chfMLE1}), the Laplace transform of $\tau = (\tau_1,\tau_2,\tau_3)$ evaluated at $t = (t_1,t_2,t_3)$, depending on the identity of the source, is given by
\begin{align*}
\varphi_u(t)
&= \varphi_a(t_1)\cdot\varphi_b(t_3)\cdot\varphi_d(t_3)\cdot\varphi_e(t_2);\\
\varphi_v(t)
&= \varphi_a(t_1)\cdot\varphi_b(t_1+t_2)\cdot\varphi_d(t_3)\cdot\varphi_e(t_2);\\
\varphi_w(t)
&= \varphi_a(t_1)\cdot\varphi_b(t_1+t_2)\cdot\varphi_c(t_1+t_2+t_3)\cdot\varphi_d(t_3)\cdot\varphi_e(t_2).
\end{align*}

Table~\ref{tab:basiclaplaces} displays the Laplace transforms of the edge delay distributions we use to evaluate our results.

\begin{table}
\begin{tabular}{|c|c|c|}
\hline
Distribution & Parameters & Laplace Transform \\
\hline
$\text{Exponential}(\lambda)$ & $\lambda>0$ & $\frac{\lambda}{\lambda+t}$ \\
\hline
$\text{PosNormal}(\mu,\sigma^2)$ & $\mu\ge0,\,\sigma>0$ & $\frac{\Phi\left((\mu/\sigma)-\sigma t\right)}{\Phi\left(\mu/\sigma\right)}\,e^{-\mu t + \frac{\sigma^2t^2}{2}}$\\
\hline
$\text{Uniform}(a,b)$ & $0\le a < b <+\infty$ & $\frac{e^{-at}-e^{-bt}}{(b-a)\,t}$\\ 
\hline
$\text{AbsCauchy}(\sigma)$ & $\sigma > 0$ & $\frac{1}{\pi}\left(2\,\text{Ci}(t\sigma)\sin(t\sigma)+\cos(t\sigma)\left(\pi-2\,\text{Si}(t\sigma)\right)\right)$\\
\hline
\end{tabular}
\caption{Some continuous distributions on the positive real line and their corresponding Laplace transforms in terms of their parameters. By PosNormal we refer to a Gaussian distribution with mean $\mu$ and variance $\sigma^2$, conditioned to be nonnegative. In the table, $\Phi$ is the cumulative distribution function of a standard normal, respectively. The AbsCauchy distribution refers to the absolute value of a Cauchy random variable with location and scale parameters $0$ and $\sigma$, respectively. In the table, $\text{Ci}$ and $\text{Si}$ are the cosine and sine integrals, respectively.}
\label{tab:basiclaplaces}
\end{table}

%%%%%%%%%%%%%%%%%%%%%%%%%%%%%%%
\section{Laplace-based Source Localization}
\label{sec:Localization}

Classical statistical inference techniques for point estimation aim to minimize the mean square error between a statistic and an unknown parameter---implicitly restricting the statistics of interest to those with finite second (and consequently first) moments. Other methods, such as maximum likelihood estimation, rely on explicit formulas for the joint distribution of the data.

In the context of observer-based source localization, however, the observers' infection times often lack explicit joint density functions or finite second moments. In situations analogous to this, some point estimation methods have exploited characteristic functions to estimate parameters~\cite{FeuMcD81, FeuMcD81b, MadSen87, Elgin2011, BerKle23}. The central idea of these methods is that the empirical characteristic function of the data converges to the characteristic function of its distribution over compact sets as the sample size increases. In particular, the parameters of the unknown distribution can be estimated by fitting the characteristic function to its empirical counterpart. This is conveyed by comparing the two functions over a grid of points in the domain. 

In this section, we adapt the latter methodology to estimate the source of infection in a tree by working with Laplace transforms instead of characteristic functions. This choice is appropriate not only because infection times are non-negative, but also due to the explicit formula for the Laplace transform of the observers' infection times given in Theorem~\ref{thm:chf0conditioning}.

If $\tau_1,\ldots,\tau_k$ are $k$ independent realizations of $\tau=(\tau_o)_{o \in \curlyO}$, the \textit{empirical Laplace transform} of $\tau$ is the function $\hat\varphi:\RR_+^\curlyO\to[0,1]$ defined as
\begin{equation}
\label{def:empiricalLaplace}
\hat\varphi(t):=\frac{1}{k}\sum_{i=1}^ke^{-\langle t,\tau_i\rangle},\text{ for }t\ge0.
\end{equation}
(For each $t\ge0$, $\hat\varphi(t)$ is an unbiased estimator of $\varphi(t)$.) \textit{In our setting, however, $k=1$ as we have a single observation of the infection times of the observer nodes in $T$.} We address this additional challenge in Section~\ref{sec:augmentation} and for now assume that $k\ge1$ is fixed.

In the traditional approach, one selects a grid of values $t_1, \ldots, t_n\in\RR_+^\curlyO$ and estimates the source by minimizing over $v \in V \setminus \curlyO$ the quantity
\[\sum_{i=1}^k\sum_{j=1}^n\big(\hat\varphi(t_j)-\varphi_v(t_j)\big)^2.\]
This approach, however, has the disadvantage of not being scale-invariant: if the units of time are changed by a constant factor (e.g., measuring time in weeks instead of days), the source estimator may also change. To address this issue, we fix a $2\le p\le+\infty$ and instead aim to solve the optimization problem
\begin{equation}
\label{ide:opt.prob.}
\min_{v\in V\setminus\curlyO}\|\hat\varphi-\varphi_v\|_p,
\end{equation}
where $\|\cdot\|_p$ denotes the $L^p$-norm on $\RR_+^\curlyO$ with respect to the Lebesgue measure. (Weighted $L^p$-norms may also be used, provided the weighting function is homogeneous to keep the source estimator scale-invariant.) 

Since $\hat\varphi$ is almost surely a linear combination of functions in $L^p$,  $\hat\varphi\in L^p$. On the other hand, the Laplace transform is a continuous linear operator from $L^q$ to $L^p$, where $q := p / (p - 1)\in[1,2]$ is interpreted as $1$ when $p$ is infinity~\cite{Pea22}. In particular, if the probability density function of $\tau$ is in $L^q$, then the objective function above is finite for each $v\in V\setminus\curlyO$. Unfortunately, however, for $2 \le p < +\infty$, computing the $L^p$-norm in (\ref{ide:opt.prob.}) is computationally demanding, particularly in high dimensions. Moreover, since in general we can only assert that the p.d.f. of $\tau$ lies in $L^1$, selecting $p = +\infty$ is a natural choice. Accordingly, we propose estimating the source by solving the following optimization problem:
\begin{equation}
\label{def:shatGENERAL}
\hat s
:=\argmin_{v\in V\setminus\curlyO}\|\hat\varphi-\varphi_v\|_\infty
=\argmin_{v\in V\setminus\curlyO}\,\sup_{t\in\mathbb{R}^{|\curlyO|}_+}\left|\hat\varphi(t)-\varphi_v(t)\right|.
\end{equation}
We call this the \textit{source-hat estimator}.

%%%%%%%%%%%%%%%%%%%%%%%%%%%%%%%
\subsection{Alternative Source Estimator}
\label{sec:augmentation}

We address now how to improve the source estimator in (\ref{def:shatGENERAL}) when $k = 1$, i.e., when we have a single realization of the vector of observer infection times $\tau$. A drawback of this approach is that it requires explicit expressions for certain conditional Laplace transforms. In this regard, the main result of this section (Theorem~\ref{thm:magic}) provides such a formula, albeit in terms of convolution operators, which may still be challenging to compute explicitly in practice. Nonetheless, it enables the derivation of explicit expressions for conditional Laplace transforms in networks with Exponential delays.

Guided by Theorem~\ref{thm:stattrick} in the Appendix, we can estimate the conditional Laplace transform of $\tau=(\tau_o)_{o\in\curlyO}$ given that $s=v$ by
\begin{align}
\label{def:varphihatLOWVAR}
\check\varphi_v(t)
&:=\frac{|\curlyO|-1}{2|\curlyO|-1}e^{-\langle t,\tau\rangle}+\frac{1}{2|\curlyO|-1}\sum_{o\in\curlyO}\varphi_v(t|\tau_o),
\end{align}
where
\begin{equation}
\label{def:varphihatATvLOWVAR}
\varphi_v(t|\tau_o):=\EE\left(\left.e^{-\langle t, \tau \rangle}\right|\tau_o,s=v\right).
\end{equation}
This leads us to the following alternative source estimator:
\begin{equation}
\label{def:shatLOWVAR}
\check s
:=\argmin_{v\in V\setminus\curlyO}\|\check\varphi_v-\varphi_v\|_\infty
=\argmin_{v\in V\setminus\curlyO}\,\sup_{t\in\mathbb{R}^{|\curlyO|}_+}\left|\check\varphi_v(t)-\varphi_v(t)\right|.
\end{equation}
We call this the \textit{source-check estimator}.

Importantly, while $\hat\varphi$ and $\check\varphi_v$ are both unbiased estimators of $\varphi_v$ when $s = v$, the variance of the latter can never exceed that of the former. In particular, source estimation based on the optimization in (\ref{def:shatLOWVAR}) should be preferred over that in (\ref{def:shatGENERAL})---provided that $\check\varphi_v$ is computationally tractable for each $v\in V\setminus\curlyO$.

The conditional Laplace transform in (\ref{def:varphihatATvLOWVAR}) can be made more explicit by following a similar line of reasoning to that used in the proof of Theorem~\ref{thm:chf0conditioning}, as stated next (proof omitted).

\medskip

\begin{corollary}
\label{cor:chfconditioning}
For all $v\in V\setminus\curlyO$ and $o\in\curlyO$:
\begin{align*}
\varphi_v(t|\tau_o)
\label{ide:Key4VarReduction} &=\EE\left(\left.\prod_{e\in[v,o]}e^{-\tau_e\cdot\sum\limits_{o'\in \curlyO(e|v)}\!\!\!\!t_{o'}}\right|\tau_o,s=v\right)\cdot\prod_{e\notin[v,o]}\varphi_e\!\!\left(\sum_{o\in \curlyO(e|v)}\!\!\!\!t_o\right).
\end{align*}
\end{corollary}

For instance, for the infection tree in Figure~\ref{fig:chartheoexample}, the corollary implies that
\begin{align*}
\varphi_u(t|\tau_3)
&= \varphi_a(t_1)\cdot\varphi_e(t_2)\cdot e^{-t_3\tau_3};\\
\varphi_v(t|\tau_3)&= \varphi_a(t_1)\cdot\varphi_b(t_1+t_3)\cdot\varphi_e(t_2)\cdot e^{-t_3\tau_3};\\
\varphi_w(t|\tau_3)&=\varphi_a(t_1)\cdot\varphi_b(t_1+t_2)\cdot\varphi_e(t_2)\cdot\EE\left(\left.e^{-(t_1+t_2+t_3)\tau_c-t_3\tau_d}\right|\tau_c+\tau_d,s=w\right);
\end{align*}
where, for the first and last identity above, we have used that $\tau_3=(\tau_b+\tau_d)$ when $\tau_3=(\tau_c+\tau_d)$ when $s=u$ and $s=w$, respectively. 

The explicit formulas in the first two examples above are uncommon, whereas the third is a more typical albeit simple example of the type of conditional expectations required to compute conditional Laplace transforms of the form given in Corollary~\ref{cor:chfconditioning}. The following result provides a general formula for conditional Laplace transforms of this type, which rely on the convolution operator. 

For each $c\ge0$, define the the $L^1$-endomorphism:
\[(\curlyL_cf)(x):=e^{-c x}f(x),\,x\ge0.\]

\begin{theorem}
Let $k\ge2$ be an integer. If $c_1,\ldots,c_k\ge0$ are given real numbers, and $\tau_1,\ldots,\tau_k\ge0$ are independent continuous random variables with p.d.f.'s $f_1,\ldots,f_k$, respectively, then 
\begin{equation}
\EE\left(\left.e^{-\sum\limits_{i=1}^kc_i\tau_i}\right|\sum_{i=1}^k\tau_i=t\right)
=\frac{(\curlyL_{c_1}f_1\ast\cdots\ast \curlyL_{c_k}f_k)(t)}{(f_1\ast\cdots\ast f_k)(t)},\,\text{ for all }t\ge0.
\label{ide:magic}
\end{equation} 
\label{thm:magic}
\end{theorem}

We can provide a comparatively explicit formula for equation (\ref{ide:magic}) when $\tau_1,\ldots,\tau_k$ i.i.d. exponential random variables.

\medskip

\begin{corollary}
Let $k\ge2$ be an integer. If $c_1,\ldots,c_k\ge0$ are constants, and $\tau_1,\ldots,\tau_k$ i.i.d. $\text{Exponential}(\lambda)$ random variables, then
\begin{equation}
\EE\left(\left.e^{-\sum\limits_{i=1}^kc_i\tau_i}\right|\sum_{i=1}^k\tau_i=t\right)
=g(t)\,t^{k-1}e^{-\lambda t}\,(k-1)!\cdot\prod_{i=1}^k\frac{1}{\lambda+c_i},\,\text{ for all }t\ge0;
\label{cor:Expo.magic}
\end{equation}
where $g(t)$ is the p.d.f. of a sum of independent exponential random variables with rates $(\lambda+c_1),\ldots,(\lambda+c_k)$, respectively.
\end{corollary}

\begin{proof}
Let $f_i$ and $\varphi_i$ be the p.d.f. and Laplace transform of $\tau_i$, respectively. Let $X_1,\ldots,X_k$ be independent random variables such that, for each $1\le i\le k$, $X_i$ has p.d.f. $g_i:=\curlyL_{c_i}f_i/\varphi_i(c_i)$. Then
\[\varphi_{X_i}(t)=\int_0^\infty\frac{e^{-(t+c_i)x}f_i(x)}{\varphi_i(c_i)}\,dx
=\frac{\varphi_i(t+c_i)}{\varphi_i(c_i)}=\frac{\lambda+c_i}{\lambda+c_i+t}.\]
In particular, $X_i\sim\text{Exponential}(\lambda+c_i)$, and $g:=(g_1\ast\cdots\ast g_k)$ is the p.d.f. $\sum_{i=1}^kX_i$. The corollary follows from equation (\ref{ide:nice.magic}).
\end{proof}

\begin{remark}
If $X_1,\ldots,X_k$ are independent exponentials with rates $\lambda_1,\ldots,\lambda_k>0$, respectively, then $\sum_{i=1}^k X_i$ is said to have a \textit{hypoexponential} (a.k.a. generalized Erlang distribution)  distribution. In the special case when $\lambda_i\ne\lambda_j$ for all $i\ne j$, the p.d.f. of this distribution is 
\[g(t):=\sum_{i=1}^k \lambda_i\,e^{-\lambda_i\,t}\cdot\prod_{j\ne i}\frac{\lambda_j}{\lambda_j-\lambda_i},\,t\ge0.\]
Hypoexponential distributions are particular instances of the so-called \textit{continuous phase type distribution}. In particular, if two or more of the rates $\lambda_1, \ldots, \lambda_k$ are repeated, the distribution of $\sum_{i=1}^k X_i$ is of phase type; in this case, its c.d.f. and p.d.f. can be computed using matrix exponentiation~\cite{LegJou15}.
\end{remark}

%%%%%%%%%%%%%%%%%%%%%%%%%%%%%%%
\section{Source Localization Performance and Application} \label{sec:performance}

In this section, we evaluate the performance of our source localization methods using synthetic data on usually random networks (Sections~\ref{subsec:testing1}-\ref{subsec:testing2}) as well as synthetic data on a real river network (Section~\ref{subsec:testing3}). Our tests in Section~\ref{subsec:testing1} use the source estimator in equation~(\ref{def:shatGENERAL}), whereas those in Section~\ref{subsec:testing2} use the one in~(\ref{def:shatLOWVAR}). In both sections, the observers are chosen as a subset of the leaves to avoid the difficulties discussed at the end of Section~\ref{sec:Sufficient}. This is not the case for the data in Section~\ref{subsec:testing3}, however, the localization problem can be reduced to the previous case using the tools in Section~\ref{sec:Sufficient}.

We consider synthetic infection networks with i.i.d. edge-delays drawn from the following distributions (see Table~\ref{tab:basiclaplaces}):
\begin{itemize}
\item $\text{PosNormal}(1,\sigma^2)$, with $\sigma^2 = 1/16,\ 1/4,\ 1$
\item $\text{Exponential}(1)$;
\item $\text{Uniform}(0,2)$;
\item $\text{AbsCauchy}(1)$.
\end{itemize}
For the first three of these distributions, the edge-distance between an observer and the source is proportional to the observer's expected infection time---and exactly equal to it for the middle two distributions. This does not hold for the fourth distribution, which has infinite moments of all orders. However, we selected $\sigma = 1$ because it gives a distribution with median 1. We note that, because our methods are scale invariant, their performance under i.i.d. Exponential delays, or Uniform delays anchored at $0$, does not depend on the parameters of these distributions.

Each of the above distributions reflects characteristics of practical or theoretical interest. Indeed, when the variance is small relative to the mean, the positive Normal distribution models delays resulting from the aggregation of multiple independent and short-lived delays due to the various ways in which the standard Central Limit Theorem may emerge. The Exponential is well-suited for modeling Markovian (i.e., memoryless) delays, while the Uniform distribution serves as a paradigm for high-entropy delays; in particular, Uniform delays offer minimal information about the location of a source. Finally, due to the heavy tail of the Cauchy distribution, anomalously high edge-delays are likely to occur along long paths connecting a node to the source, making localization particularly challenging.

%%%%%%%%%%%%%%%%%%%%%%%%%%%%%%%
\subsection{Hat-estimator Performance on Synthetic Networks}
\label{subsec:testing1}

In this section, we test the hat-estimator as defined in equation (\ref{def:shatGENERAL}). 

To begin, we consider a path tree with an observer at its left end (labeled $0$) and ten potential sources (labeled $1,\ldots,10$) to its right---see the top of Figure~\ref{plot:path.tree}. This simple network is well-suited for testing our methodology because---except under AbsCauchy edge-delays---the variance of the observer's infection time increases proportionally with its edge-distance from the true source.

\begin{figure}[t]
\centering
\begin{tikzpicture}[auto]
\begin{scope}[scale = 1, 
observernode/.style={circle, draw, fill=white, inner sep = 0pt, minimum size = 10pt}, 
yellownode/.style={circle, draw, fill=gray!50, inner sep = 0pt, minimum size = 10pt}]
    \node [observernode] (o) at (0,0)  {$\scriptstyle{0}$};
    \node [yellownode] (1) at (0.75,0)  {$\scriptstyle{1}$};
    \node [yellownode] (2) at (1.5,0)  {$\scriptstyle{2}$};
    \node [yellownode] (3) at (2.25,0)  {$\scriptstyle{3}$};
    \node [yellownode] (4) at (3,0)  {$\scriptstyle{4}$};
    \node [yellownode] (5) at (3.75,0)  {$\scriptstyle{5}$};
    \node [yellownode] (6) at (4.5,0)  {$\scriptstyle{6}$};
    \node [yellownode] (7) at (5.25,0)  {$\scriptstyle{7}$};
    \node [yellownode] (8) at (6,0)  {$\scriptstyle{8}$};
    \node [yellownode] (9) at (6.75,0)  {$\scriptstyle{9}$};
    \node [yellownode] (10) at (7.5,0)  {$\scriptstyle{10}$};
\end{scope}
\begin{scope}[every edge/.style={draw=black}]
    \draw (o) edge node {} (1);
    \draw (1) edge node {} (2);
    \draw (2) edge node {} (3);
    \draw (3) edge node {} (4);
    \draw (4) edge node {} (5);
    \draw (5) edge node {} (6);
    \draw (6) edge node {} (7);
    \draw (7) edge node {} (8);
    \draw (8) edge node {} (9);
    \draw (9) edge node {} (10);
\end{scope}
\end{tikzpicture}

\begin{tabular}{ccc}
\includegraphics[width=0.33\linewidth]{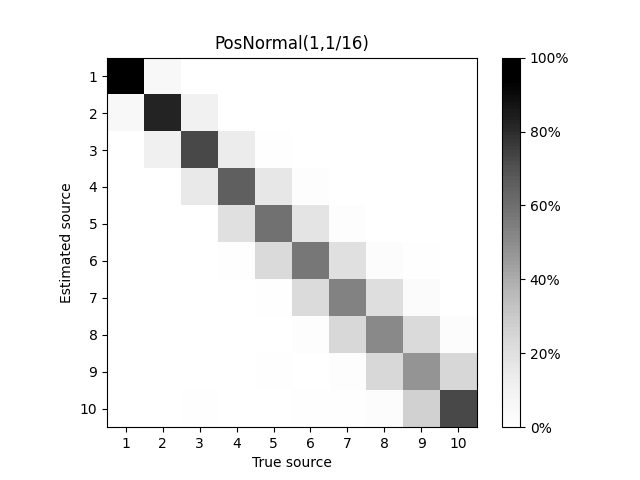} &
\includegraphics[width=0.33\linewidth]{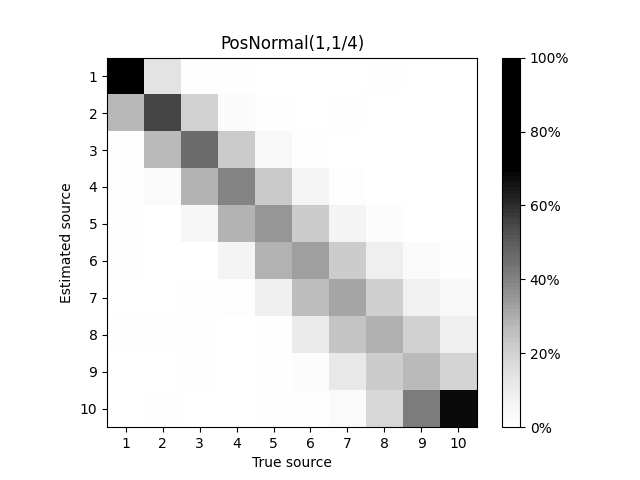}  &
\includegraphics[width=0.33\linewidth]{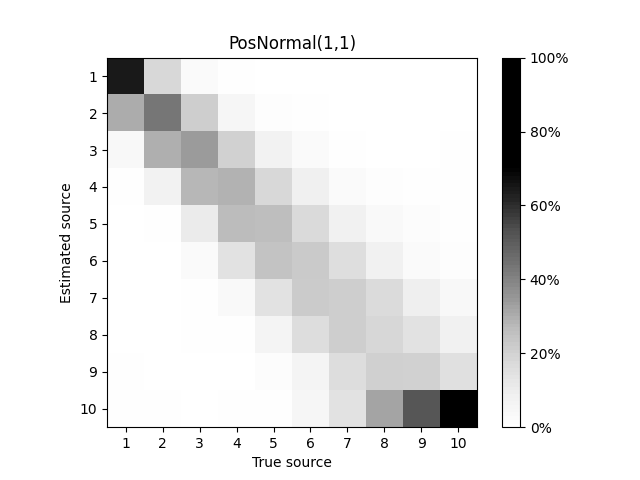} \\
\includegraphics[width=0.33\linewidth]{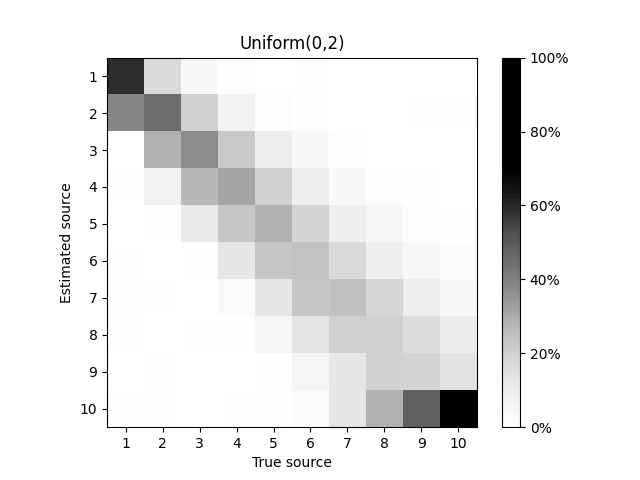} & 
\includegraphics[width=0.33\linewidth]{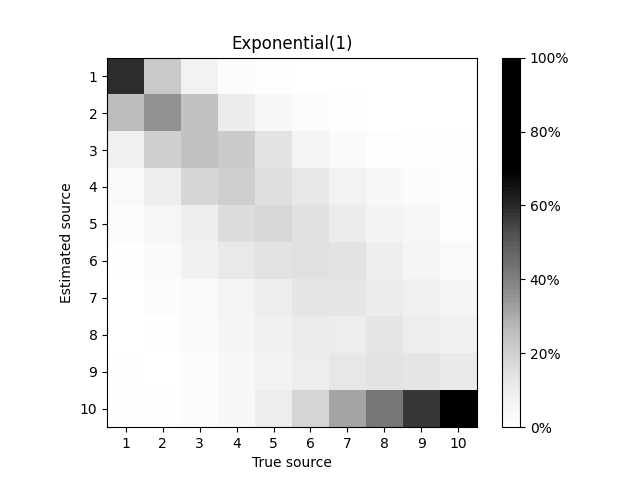}  &
\includegraphics[width=0.33\linewidth]{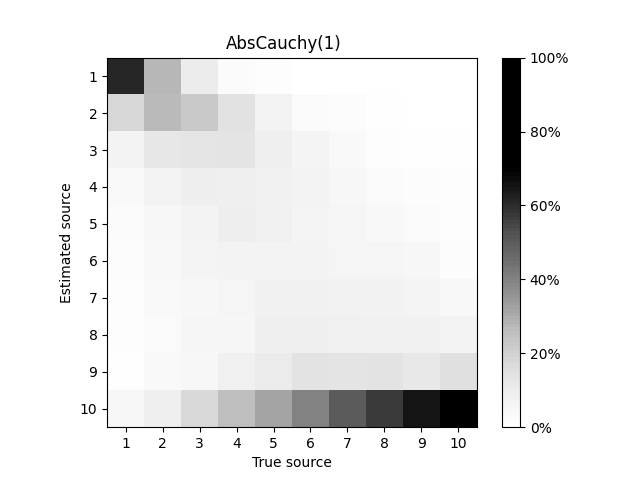}\\
\end{tabular}
\caption{Diagram of a path infection network with a single observer (top row) and confusion matrices for source localization based on the $\hat s$ estimator when using i.i.d. PosNormal (middle row), Uniform (bottom left), Exponential (bottom center), and AbsCauchy (bottom right) edge delay distributions. Each of these was run with 1,000 samples for each possible true source. The darker the shading along the diagonals and the lighter the shading off them, the better the source localization performance.}
\label{plot:path.tree}
\end{figure}
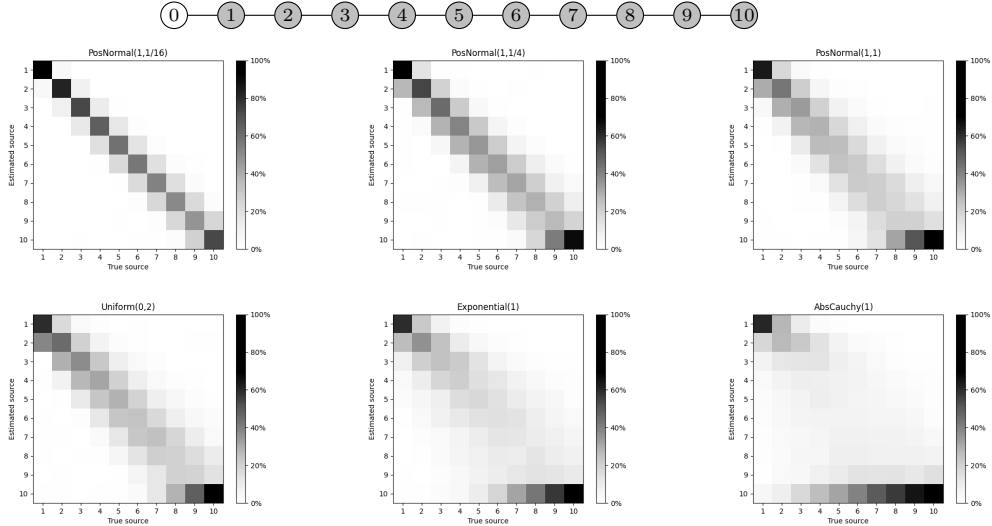

As seen in Figure~\ref{plot:path.tree}, the confusion matrices corresponding to the first two PosNormal distributions, as well as the Exponential and Uniform distributions, are mostly concentrated around the diagonal. This indicates that $\hat s$ often correctly identifies $s$ or a nearby node. In contrast, the performance deteriorates dramatically for the third PosNormal and the Absolute Cauchy distribution. The underperformance of the PosNormal distribution with $\mu = \sigma^2 = 1$ may be attributed to the rapidly increasing coefficient of variation (i.e., the ratio of standard deviation to mean) as the source moves farther from the observer. On the other hand, the heavy tail of the AbsCauchy distribution results in a high probability of anomalously large edge-delays between the observer and the source, especially when the source is distant from the observer.

To evaluate the effectiveness of the $\hat{s}$ estimator in more general infection networks, we conducted two types of experiments on random trees ranging in size. These random trees were selected uniformly at random from the set of all trees with $n$ nodes. This was done by generating Pr\"ufer sequences~\cite{Pru18} uniformly at random and then building the related trees. (Pr\"ufer sequences of length $(n-2)$ are in bijection with trees containing $n$ nodes.) All observers were selected to lie on the leaves to avoid any issues with the star arrangement configurations discussed earlier.

In the first type of experiment, we fixed the number of observers at 2 while increasing the network size, which resulted in an observer density ranging from 20\% to 2\%. As shown on the left of Figure~\ref{fig:shat.larger.trees}, the average edge-distance between $\hat{s}$ and $s$ increased sub-linearly while the standard deviation was approximately within the range of a single network edge. In contrast, in the second type of experiment, we fixed the network size at 100 nodes and increased the observer density from 1\% to 40\%. As shown on the right of Figure~\ref{fig:shat.larger.trees}, the average edge-distance between $\hat{s}$ and $s$ decreased sub-linearly, while the standard deviation again remained within the range of a single edge.

\begin{figure}[t]
    \centering
    \includegraphics[width=0.49\linewidth]{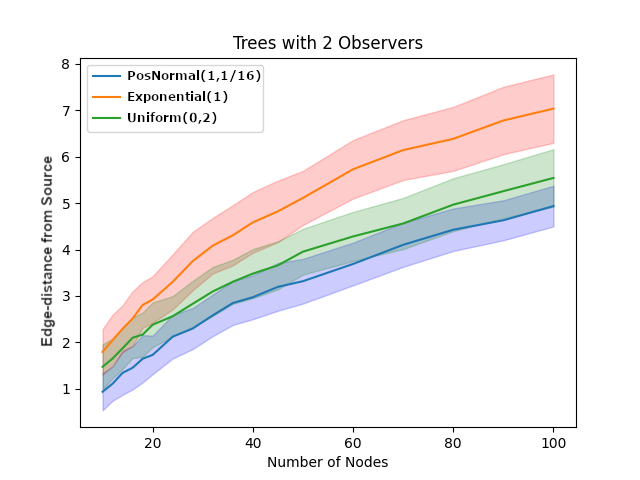}
    \includegraphics[width=0.49\linewidth]{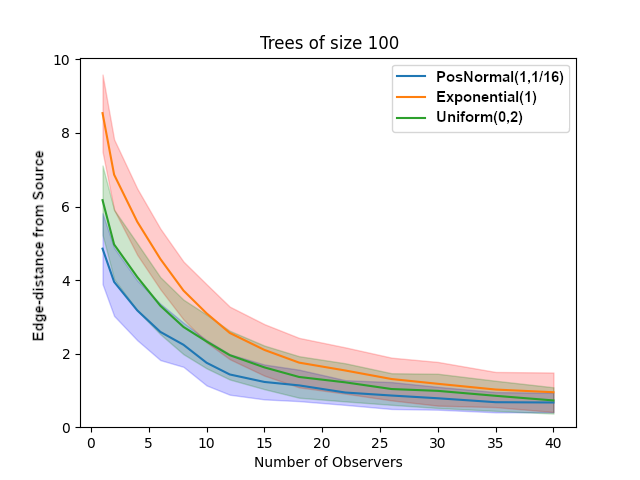}
    \caption{Left: Average edge-distance (i.e., number of edges) between $\hat{s}$ and $s$ in infection trees with only 2 observers, as the size of randomly generated trees increases. Each tree size had 1,000 samples. Right: Average edge-distance in randomly generated trees with 100 nodes, as the number of observers increases. Each number of observers had 1,000 samples. In all the plots, the shaded bands represent $\pm$ one standard deviation from the mean.}
    \label{fig:shat.larger.trees}
\end{figure}

Next, we explored how does the edge-distance between $s$ and $\hat s$ compare to the diameter of the tree (i.e., largest edge-distance between a pair of nodes in the network). As seen on Figure~\ref{fig:normalized.obs}, the average diameter-normalized edge-distance between $s$ and $\hat{s}$ becomes essentially constant for each of the three edge delay distributions tested as the observer density decreases by holding the number of observers fixed at 2. This is somewhat expected because average edge-distance between the observers and a randomly placed source should grow proportionally with the network size. 

\begin{figure}[t]
    \centering
\includegraphics[width=0.49\linewidth]{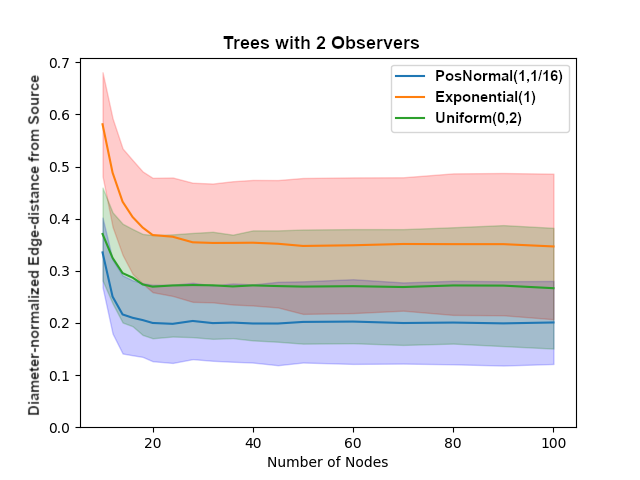} %\includegraphics[width=0.49\linewidth]{plot?}
    \caption{Left: Performance of the method normalized by the diameter of the tree in a tree with 2 observers vs. the size of the tree for uniformly at random generated trees with i.i.d.  normal, exponential, and uniform edge delay distributions. Each node-size was run with 1,000 samples.}
    \label{fig:normalized.obs}
\end{figure}

%%%%%%%%%%%%%%%%%%%%%%%%%%%%%%%
\subsection{Hat-estimator Performance in a River Network}
\label{subsec:testing3}

\begin{figure}[t]
\centering
\includegraphics[width=0.67\linewidth]{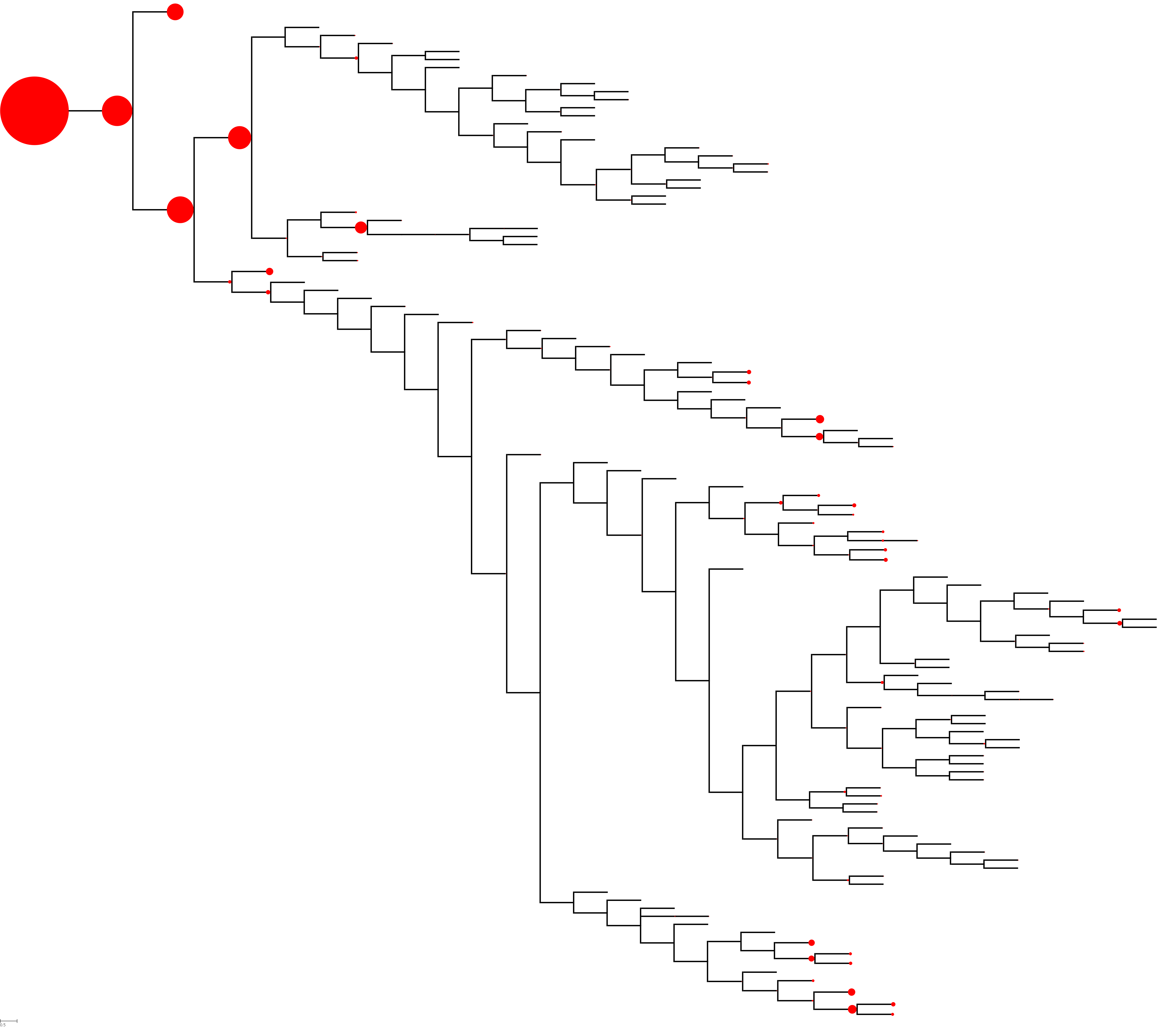} 
\includegraphics[width=0.1\linewidth]{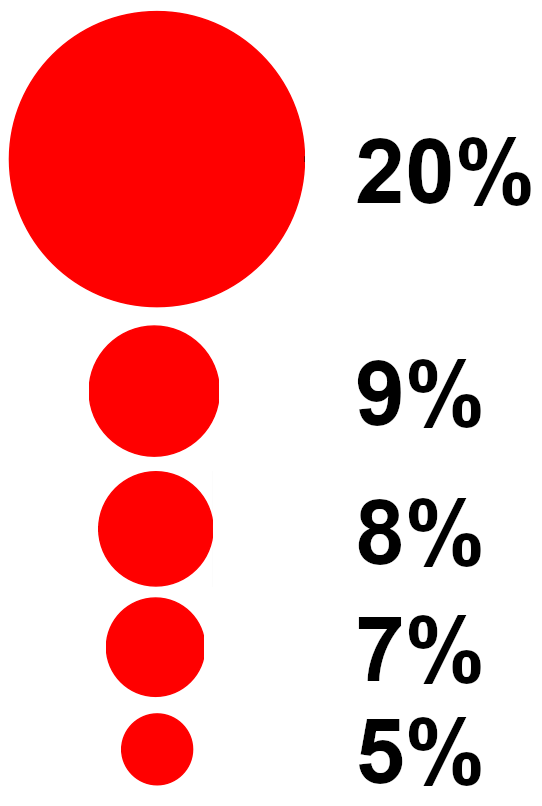} 
\caption{Heatmap of the empirical probability that each node in the river network is identified as the source when the infection originates at the root of the periodogram. Larger nodes correspond to those more frequently predicted as the source by the $\hat{s}$ estimator across 1,000 trials.} 
\label{Thukela.river}
\end{figure}

In real-world settings, infection trees are rarely uniformly distributed over the set of all possible trees and often exhibit structural features shaped by factors such as geography, contact patterns, or transmission dynamics. 

To assess the performance of the hat-estimator in a more realistic scenario, we consider an infection network from a cholera outbreak in the KwaZulu-Natal province of South Africa in the year 2000. This epidemic was caused by a strain of \textit{Vibrio cholerae}, which typically spreads through aquatic environments---in this case, along the Thukela River basin. Because the infection followed a river system, the resulting network naturally forms a directed tree-like structure. 

This network is considered in the context of source localization by~\cite{PRL}, where edge-delays were modeled using Normal distributions with parameters estimated by~\cite{bertuzzo2010spatially, bertuzzo2008space}, who modeled infection propagation with a system of differential equations. 

Here, in each trial, the source was set to the root of the network, and three observers were selected uniformly at random, excluding the root. For the edge-delays, we reused the parameters in~\cite{PRL}, but assumed the delays follow Positive Normal distributions. This adjustment has a negligible impact on the original model, since the probability mass below zero is marginal.

We emphasize that the directional flow of water along the river is still compatible with our methods, which were developed for undirected networkz. However, because the infection must originate upstream of any infected observer, only nodes simultaneously upstream of all observers can be the source. As this would sharply reduce the set of candidate sources in each trial and thus trivialize our performance test, we instead assume the river network is undirected---or, equivalently, that the infection can propagate both downstream and upstream from the source.

As seen in Figure~\ref{Thukela.river}, our method identifies the true source in a significant fraction of trials. Moreover, with only 3 observers placed at random among the 246 nodes in the river basin, approximately 50\% of the estimates fall within the five nodes nearest (in terms of edge-distance) to the true source. These same nodes are also the most frequently identified as the source and represent only about 2\% of all nodes where the infection could have originated.

%%%%%%%%%%%%%%%%%%%%%%%%%%%%%%%
\subsection{Check-estimator Performance under Markovian Delays}
\label{subsec:testing2}

In this section, we test the check-estimator as defined in equation (\ref{def:shatLOWVAR}) and compare its performance to that of the hat-estimator of the infection source. We recall that the former estimator relies on formulas for conditional Laplace transforms, which we determined explicitly only for exponential delays. Nevertheless, this class of edge delay distributions may be well-suited to settings in which information is transmitted through a network in a reasonably memoryless manner.

As seen in Figure~\ref{fig:s.check.magic}, the $\check s$-estimator performs on average marginally better than the $\hat s$-estimator in terms of edge-distance to the true source, in both low- and high-observer-density scenarios. However, as seen in the same plots, the standard deviation of the edge-distance between $\check{s}$ and $s$ is often markedly lower than that of $\hat{s}$. This feature is consistent with the variance reduction technique (see Section~\ref{sec:low.var.s.check}) that motivated the definition of the source-check estimator in (\ref{def:shatLOWVAR}). Thus, when the conditional Laplace transforms of the form in Theorem~\ref{thm:magic} can be computed explicitly, the $\check{s}$ estimator should be preferred over its precursor.

\begin{figure}[t]
\centering
\includegraphics[width=0.49\linewidth]{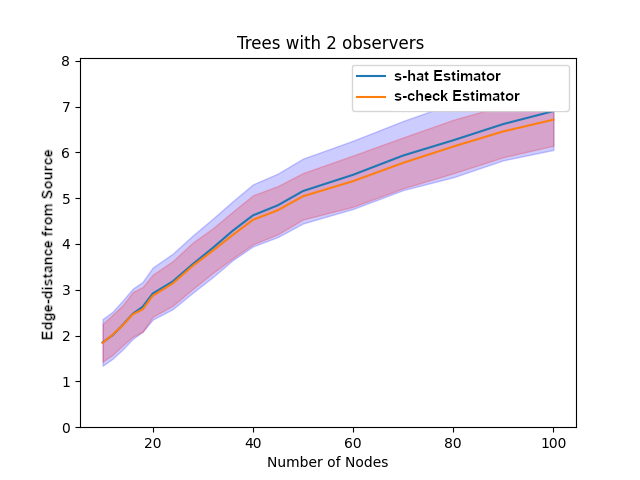}
\includegraphics[width=0.49\linewidth]{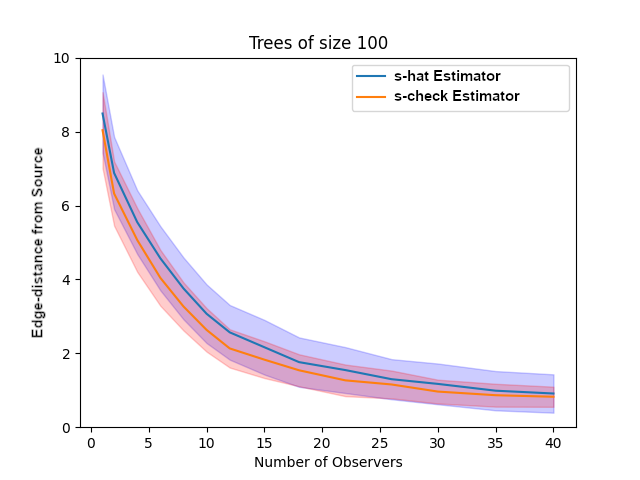}
\caption{Left: Average edge-distance between $\hat{s}$ and $s$, and $\check{s}$ and $s$, in infection trees with exponential delays and only 2 observers, as the size of randomly generated trees increases. Right: Average edge-distance between $\hat{s}$, $\check{s}$ and $s$ in randomly generated infection trees of fixed size with exponential delays, as the number of observers increases. In both plots, the shaded bands represent $\pm$ one standard deviation from the mean, estimated from 1,000 simulations at each value along the abscissa.}
\label{fig:s.check.magic}
\end{figure}

Altogether, the simulations in this and the preceding sections make a compelling case for our Laplace-derived estimators for source localization in SI networks with a tree structure. In the next section, however, we show that this structure is too restrictive for source estimation in networks with more complex topologies.

%%%%%%%%%%%%%%%%%%%%%%%%%%%%%%%
\section{Limitations on Networks with Cycles}
\label{sec:limitations}

The source localization problem on arbitrary graphs is significantly more challenging than on trees, as cycles allow infections to propagate from a source along multiple, competing, and often overlapping paths.

In our context of SI infections, this issue has been addressed by reducing the network to a spanning tree using some criteria, e.g., a breadth-first search of the network~\cite{PRL}. The rationale is that, in the absence of recovery, the infection propagates along a growing subtree, eventually becoming a spanning tree of the whole network. Here we argue, however, that \textit{even with full knowledge of the spanning tree generated by the infection, additional complications emerge that have been largely overlooked in the literature.} As we see next, these issues arise even in the simplest non-trivial infection network containing a single cycle and persist even when the edge-delays are memoryless (i.e., exponentially distributed).

Before proceeding, we recall the following well-known properties of the exponential distribution.

\medskip

\begin{lemma}
\cite[Theorem 2.1]{Dur16}. Let $\lambda_1,\ldots,\lambda_k>0$ be given and $E_1,\ldots,E_k$ be independent random variables with $E_i\sim\textit{Exponential}(\lambda_i)$. Let $I$ be the almost surely unique random index such that $E_I=\min_{i=1,\ldots,k}E_i$. Then, for each $i$:
\begin{itemize}
    \item[(a)] If $t\ge0$ then, conditioned on having $E_i\ge t$, $(E_i-t)\sim\text{Exponential}(\lambda_i)$.
    \item[(b)] $\PP(I=i)={\lambda_i}\Big/{\sum\limits_{j=1}^k\lambda_j}$. %, for each $1\le j\le k$.
    \item[(c)] $E_I\sim\text{Exponential}\left(\sum_{j=1}^k\lambda_i\right)$; and 
    \item[(d)] $E_I$ and $I$ are independent.
\end{itemize}
\label{lem:basic}
\end{lemma}

Property (a) is known as the \textit{memoryless property} of the exponential distribution. This is the only continuous probability distribution supported on $[0,+\infty)$ that is memoryless. 

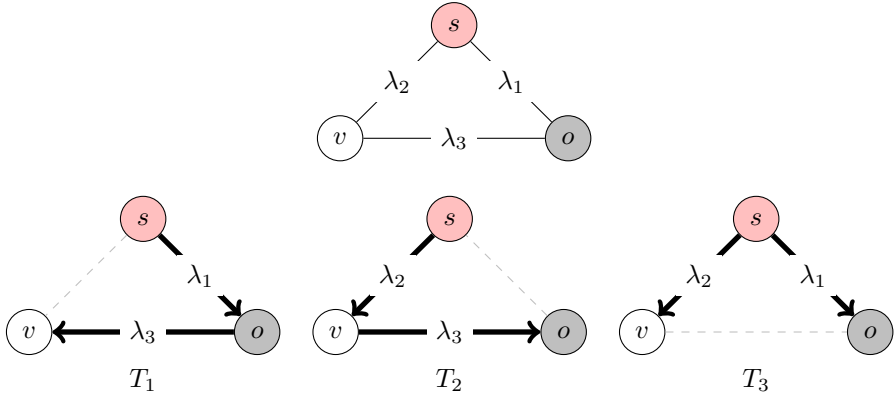
\begin{figure}[t]
\centering
\begin{tikzpicture}[scale = 1, 
rednode/.style={circle,draw, fill=red!25, inner sep = 0pt, minimum size = 17pt},
observernode/.style={circle,draw, fill=gray!50, inner sep = 0pt, minimum size = 17pt}, whitenode/.style={circle,draw, fill=white, inner sep = 0pt, minimum size = 17pt}]
% nodes
\node [rednode] (1) at (0,0) {$s$};
\node [whitenode] (2) at (-1.5,-1.5) {$v$};
\node [observernode] (3) at (1.5,-1.5) {$o$};
% edges
\draw (1) -- (2) node [midway, fill=white] {$\lambda_2$}; 
\draw (2) -- (3) node [midway, fill=white] {$\lambda_3$};
\draw (3) -- (1) node [midway, fill=white] {$\lambda_1$};
\end{tikzpicture}

\vspace{12pt}

\begin{tabular}{ccc}
%Tree T1 next
\begin{tikzpicture}[scale = 1, 
rednode/.style={circle,draw, fill=red!25, inner sep = 0pt, minimum size = 17pt},
observernode/.style={circle,draw, fill=gray!50, inner sep = 0pt, minimum size = 17pt}, whitenode/.style={circle,draw, fill=white, inner sep = 0pt, minimum size = 17pt}]
% nodes
\node [rednode] (1) at (0,0)  {$s$};
\node [whitenode] (2) at (-1.5,-1.5) {$v$};
\node [observernode] (3) at (1.5,-1.5)  {$o$};
% edges
\draw[dashed,black!25!white] (1) -- (2); 
\draw[->,line width=2pt] (3) -- (2) node [midway, fill=white] {$\lambda_3$};
\draw[->,line width=2pt] (1) -- (3) node [midway, fill=white] {$\lambda_1$};
\end{tikzpicture}
&
%Tree T2 next
\begin{tikzpicture}[scale = 1, 
rednode/.style={circle,draw, fill=red!25, inner sep = 0pt, minimum size = 17pt},
observernode/.style={circle,draw, fill=gray!50, inner sep = 0pt, minimum size = 17pt}, whitenode/.style={circle,draw, fill=white, inner sep = 0pt, minimum size = 17pt}]
% nodes
\node [rednode] (1) at (0,0)  {$s$};
\node [whitenode] (2) at (-1.5,-1.5) {$v$};
\node [observernode] (3) at (1.5,-1.5)  {$o$};
% edges
\draw[->,line width=2pt] (1) -- (2) node [midway, fill=white] {$\lambda_2$}; 
\draw[->,line width=2pt] (2) -- (3) node [midway, fill=white] {$\lambda_3$};
\draw[dashed,black!25!white] (3) -- (1); 
\end{tikzpicture}
&
%Tree T3 next
\begin{tikzpicture}[scale = 1, 
rednode/.style={circle,draw, fill=red!25, inner sep = 0pt, minimum size = 17pt},
observernode/.style={circle,draw, fill=gray!50, inner sep = 0pt, minimum size = 17pt}, whitenode/.style={circle,draw, fill=white, inner sep = 0pt, minimum size = 17pt}]
% nodes
\node [rednode] (1) at (0,0)  {$s$};
\node [whitenode] (2) at (-1.5,-1.5) {$v$};
\node [observernode] (3) at (1.5,-1.5)  {$o$};
% edges
\draw[->,line width=2pt] (1) -- (2) node [midway, fill=white] {$\lambda_2$}; 
\draw[dashed,black!25!white] (2) -- (3); 
\draw[->,line width=2pt] (1) -- (3) node [midway, fill=white] {$\lambda_1$};
\end{tikzpicture} \\
$T_1$ & $T_2$ & $T_3$
\end{tabular}
\caption{Top: Non-tree infection network with vertex set $\{s, o, v\}$ and exponential edge-delays, with rates indicated along them. Node $s$ is the true source, and $o$ is the sole observer. Bottom: Possible spanning trees in a non-tree infection network with three vertices. The edges of each spanning tree are thickened for clarity. From left to right, the trees are labeled $T_1$, $T_2$, and $T_3$, respectively.}
\label{fig:trinetwork}
\end{figure}

\medskip

Consider the triangular infection network with a single observer at the top in Figure~\ref{fig:trinetwork}, where the edge-delays are independent exponential random variables with rates $\lambda_1,\lambda_2,\lambda_3>0$, as displayed in the figure. The observer infection time is therefore
\begin{align}
\tau_o
\label{ide:tauASmin} &=\min\big\{\tau_{\{s,o\}},\tau_{\{s,v\}}+\tau_{\{v,o\}}\big\}.
\end{align}
It turns out that the distribution of $\tau_o$ is not determined solely by the marginal probability distributions of $\tau_{{s,o}}$ and $(\tau_{{s,v}}+\tau_{{v,o}})$ because their joint distribution depends on how the infection propagates through the triangular network.

To see why, let $\TT$ denote the random subtree that describes how the infection propagates in the network. This tree can be any of three spanning trees, denoted $T_1$, $T_2$, and $T_3$ (see the bottom of Figure~\ref{fig:trinetwork}).

The following result characterizes the distribution of $\TT$ and the conditional distribution of $\tau_o$ based on this infection propagation subtree. The $\oplus$ symbol is used to denote the summation of independent random variables. 

\medskip

\begin{proposition}
For the triangular infection network in Figure~\ref{fig:trinetwork}:
\[\PP(\TT=T)=
\begin{cases}
\frac{\lambda_1\lambda_3}{(\lambda_1+\lambda_2)\cdot(\lambda_2+\lambda_3)}, & T=T_1\\
\frac{\lambda_2\lambda_3}{(\lambda_1+\lambda_2)\cdot(\lambda_1+\lambda_3)}, & T=T_2\\
\frac{\lambda_1\lambda_2(\lambda_1+\lambda_2+2\lambda_3)}{(\lambda_1+\lambda_2)\cdot(\lambda_2+\lambda_3)\cdot(\lambda_3+\lambda_1)}, & T=T_3.
\end{cases}\]
Moreover, the conditional probability distribution of $\tau_o$ given $\TT$ is
\[\tau_o\big|[\TT=T]
\sim\begin{cases}
\text{Exponential}(\lambda_1+\lambda_2),& \text{ when }T=T_1\\
\text{Exponential}(\lambda_1+\lambda_2)\oplus\text{Exponential}(\lambda_1+\lambda_3),& \text{ when }T=T_2\\
\text{Exponential}(\lambda_1+\lambda_2)\oplus B\cdot\text{Exponential}(\lambda_1+\lambda_3),& \text{ when }T=T_3,
\end{cases}\]
where $B$ is an independent binary r.v. that takes the value $1$ with probability $(\lambda_2+\lambda_3)/(\lambda_1+\lambda_2+2\lambda_3)$.
\label{prop:triangle}
\end{proposition}

\medskip

According to the Proposition, none of the conditional distributions of $\tau_o$ coincide with the distribution it would have if the original network had been one of the corresponding spanning trees from the outset, as is commonly assumed in the literature. For instance, if $\TT=T_1$ then $\tau_o\sim\text{Exponential}(\lambda_1 + \lambda_2)$ rather than an $\text{Exponential}(\lambda_1)$. The distribution of $\tau_o$ is therefore a mixture of components that differ significantly from the marginal edge-delay distributions of the original model. 

Due to Kirchhoff's matrix tree theorem~\cite{Kik47}, in large networks these mixtures will have a super-exponential number of components, each depending in a cumbersome manner on the original edge-delay distributions. In particular, the common heuristic of selecting a spanning tree (either randomly or using an optimization criterion) to estimate the infection source while retaining the marginal edge-delay distributions is not theoretically sound, and new approaches should be investigated for such cases.

%%%%%%%%%%%%%%%%%%%%%%%%%%%%%%%
\section{Concluding Remarks}
\label{sec:conclusions}

We have studied theoretical aspects of identifiability and complexity in estimating the source of infection in undirected tree networks, where only a subset of nodes (the observers) report their infection times. Our methods rely on the joint Laplace transform of these times rather than on their joint probability density, which is often intractable.

We have assumed that each observer reported only a single infection time. This is realistic at the onset of biological epidemics, but it makes accurate source estimation considerably more difficult. Nevertheless, our methods can be directly extended to situations with multiple vectors of observer infection times, for example, when a hidden bad actor repeatedly spreads misinformation on a social network.

Our methods are scale-invariant and apply to any contagion model between neighboring nodes, provided that the transmission delays of infection along edges are (probabilistically) independent and admit explicit Laplace transforms. In particular, they cover a wide range of edge delay models, including mixed ones, beyond the well-studied case of Gaussian delays.

We tested our methods across a wide range of networks and edge-delay models while varying the observers' relative proportion. On average, for our first method (source-hat estimator), the edge-distance between the estimator and the source varied sub-linearly with the observers' density---rising as the density decreased and falling as it increased. Our results improved with our second method (the source-check estimator), which we tested on networks with exponential (memoryless) edge-delays. 

Finally, we highlighted often-overlooked technical issues in extending tree-based source localization methods to general graphs, i.e., networks with cycles that permit many---often exponentially many---infection paths from the source to each observer. Substantial challenges remain for such networks, as even the first moments (e.g., expectations and variances) of the observers’ infection times are difficult to characterize. 

%%%%%%%%%%%%%%%%%%%%%%%%%%%%%%%
\section{Technical Proofs and Auxiliary Results}
\label{sec:proofs}

%%%%%%%%%%%%%%%%%%%%%%%%%%%%%%%
\subsection{Proof of Theorem~\ref{thm:equiv}}

Since $T$ is connected, for any equivalence class $r$ and $o_1\in \curlyO$, there exists $o_2\in\partial r$ such that $o_1\in V_{o_2;r}$. But, if $r$ is feasible, then for all $o\in\partial r$ and $\omega\in V_{o,r}$, we have $\tau_o\le\tau_\omega$, with equality only if $\omega=o$. Hence, $\tau_o$, for $o\in \curlyO$, must be minimized at some $o\in\partial r$.

To show the converse, suppose that $\omega = \argmin_{o\in \curlyO} \tau_o \in \partial r$ but that $r$ is not feasible. Let $o\in \partial r$ and $o_1,o_2\in V_{o,r}$ be such that $o_2$ descends from $o_1$ in $T_{o,r}$ but $\tau_{o_2} < \tau_{o_1}$.

If $s\notin V_{o;r}$, the only way the infection can reach $o_2$ is by first infecting $o_1$, contradicting the assumption that $\tau_{o_1} > \tau_{o_2}$. Hence, $s\in V_{o;r}$. However, to infect $\omega$, the infection must first reach $o$, which is only possible if $o = \omega$; otherwise, $\omega$ could not have the smallest infection time among the observers.

Let $s\wedge o_2$ be the least common ancestor of $s$ and $o_2$ in $T_{o;r} \,(=T_{\omega;r})$. In particular, we have $s\wedge o_2\in[o_2,o]=[o_2,o_1]\cup[o_1,o]$. Since $s\wedge o_2\in[o_1,o]$ is not possible because $\tau_{o_2}<\tau_{o_1}$, it must be the case that $s\wedge o_2\in[o_2,o_1]\setminus\{o_1\}$. But then $\tau_{s\wedge o_2}<\tau_{o_1}\le\tau_o=\tau_\omega$, which is again not possible. Consequently, $r$ must be feasible, completing the proof of the theorem.

%%%%%%%%%%%%%%%%%%%%%%%%%%%%%%%
\subsection{Proof of Theorem~\ref{thm:sufficiency}}

Let $T_R=(V_R,E_R)$ be the subgraph of $T$ with vertex set $V_R=\cup_{r\in R}(r\cup\partial r)$. Since $R$ is a star arrangement, $T_R$ is a subtree of $T$. Moreover, since $s\in V_R$, the joint distribution of $\tau_o$, $o\in\partial R$, is solely determined by the delays $\tau_e$, with $e\in E_R$.

The sets $V_{o,R}\cap \curlyO$, $o\in\partial R$, partition $\curlyO$. Further, if $o\in\partial R$ and $\omega\in V_{o,R}$ then $\tau_w=\tau_o+\sum_{e\in[o,\omega]}\tau_e$. But $[o,\omega]\subset E\setminus E_R$ and, for each $e\in E\setminus E_R$, $s\notin e$. Hence, the random variables $(\tau_{\omega} - \tau_o)$, with $o \in \partial R$ and $\omega \in V_{o,R}$, are independent of $\tau_o$, $o \in \partial R$, and their joint distribution remains the same regardless of the identity of the source node in $V_R$; which shows the theorem. 

%%%%%%%%%%%%%%%%%%%%%%%%%%%%%%%
\subsection{Improving Single Multidimensional-Sample Estimation}
\label{sec:low.var.s.check}

Let $X=(X_1,\ldots,X_d)$ be a random vector and $F:\mathbb{R}^d\to\mathbb{R}$ a given function. Define $\theta=\mathbb{E}\big(F\big)$; in particular, $F:=F(X)$ is an unbiased statistic for $\theta$. 

Next, we see how to construct from $F(X)$ an unbiased statistic for $\theta$ but of a smaller variance provided that, on average, the conditional variance of $F$ given any $X_i$ is comparable to that of $F$ without conditioning. The modified statistic resembles the H\'ajek projection of $F(X)$~\cite{Van98}, although the latter would assume that $X_1, \ldots, X_d$ are independent.

\medskip

\begin{theorem}
\label{thm:stattrick}
Assume that $\EE(F^2)<+\infty$ and $\EE(F)=\theta$. Define
\begin{equation}
\label{def:G}
G:=\frac{d-1}{2d-1}F+\frac{1}{2d-1}\sum_{i=1}^d\EE(F|X_i);\text{ in particular},\EE(G)=\theta.
\end{equation}
If $\alpha\ge0$ is such that $\EE\left(\frac{1}{d}\sum\limits_{i=1}^d\VV(F|X_i)\right)\ge\alpha\cdot\VV(F)$, then 
\begin{equation}
\label{ine:GalphaF}
\VV(G)
\le\left(1-\frac{\alpha\,d}{2d-1}\right)\VV(F).
\end{equation}
\end{theorem}
\begin{remark}
$\alpha\le1$ because $\EE\big(\VV(F|X_i)\big)\le\VV(F)$ for each $i$. In particular, $\VV(G)\le\VV(F)$.
\end{remark}

\begin{proof}
Consider $0\le\lambda\le1$ to be selected later, and define
\[G:=\lambda F+\frac{1-\lambda}{d}\sum_{i=1}^d\EE(F|X_i).\]
The statistic in (\ref{def:G}) corresponds to $\lambda=(d-1)/(2d-1)$, which we will see is optimal for the inequality in (\ref{ine:GalphaF}). 

For the sake of a simpler notation, let $E_i:=\EE(F|X_i)$ and $V_i:=\VV\big(F|X_i)$. Then
\begin{align*}
\VV(G)
&=\lambda^2\VV(F)+\frac{2\lambda(1-\lambda)}{d}\sum_{i=1}^d\cov\big(F,E_i\big)+\frac{(1-\lambda)^2}{d^2}\sum_{i,j=1}^d\cov\big(E_i,E_j\big) \\
&=\lambda^2\VV(F)+\frac{2\lambda(1-\lambda)}{d}\sum_{i=1}^d\cov\big(F,E_i\big)+\frac{(1-\lambda)^2}{d^2}\sum_{i=1}^d\VV\big(E_i\big) \\
&\qquad\qquad\qquad\qquad\qquad\qquad\qquad\qquad\qquad\qquad\qquad
+\frac{(1-\lambda)^2}{d^2}\sum_{i\ne j}^d\cov\big(E_i,E_j\big).
\end{align*}
But $\cov\big(F,E_i\big)=\VV(E_i)=\VV(F)-\EE\big(V_i\big)$ because 
$\EE(\cdot|X_i)$ can be regarded an orthogonal projection onto the linear space of measurable transformations of $X_i$ with finite second moment, and the \textit{conditional variance formula}. In particular, due to the \textit{Cauchy-Schwarz inequality}:
\[\cov\big(E_i,E_j\big)
\le\sqrt{\VV(E_i)\cdot\VV(E_j)}
\le\VV(F).\]
Therefore
\begin{align*}
\VV(G)
&=\left(\lambda^2+2\lambda(1-\lambda)+\frac{(1-\lambda)^2}{d}\right)\VV(F) 
- \left(\frac{2\lambda(1-\lambda)}{d}+\frac{(1-\lambda)^2}{d^2}\right)\sum_{i=1}^d\EE(V_i) \\
&\qquad\qquad\qquad\qquad\qquad\qquad\qquad\qquad\qquad\qquad\qquad\quad
+\frac{(1-\lambda)^2}{d^2}\sum_{i\ne j}^d\cov\big(E_i,E_j\big) \\
&=\left(1-\frac{(d-1)(1-\lambda)^2}{d}\right)\VV(F)-\left(2\lambda+\frac{1-\lambda}{d}\right)\cdot\frac{1-\lambda}{d}\sum_{i=1}^d\EE(V_i) \\
&\qquad\qquad\qquad\qquad\qquad\qquad\qquad\qquad\qquad\qquad\qquad\quad
+\frac{(1-\lambda)^2}{d^2}\sum_{i\ne j}^d\cov\big(E_i,E_j\big) \\
&\le\left(1-\frac{(d-1)(1-\lambda)^2}{d}\right)\VV(F)-\alpha(1-\lambda)\left(2\lambda+\frac{1-\lambda}{d}\right)\VV(F) \\
&\qquad\qquad\qquad\qquad\qquad\qquad\qquad\qquad\qquad\qquad\qquad\quad
+\frac{(1-\lambda)^2}{d}(d-1)\VV(F)\\
&=\left(1-\alpha(1-\lambda)\frac{1+(2d-1)\lambda}{d}\right)\VV(F);
\end{align*}
and a simple calculation shows that the factor multiplying $\VV(F)$ above is minimized at $\lambda=(d-1)/(2d-1)$.
\end{proof}

%%%%%%%%%%%%%%%%%%%%%%%%%%%%%%%
\subsection{Proof of Theorem~\ref{thm:magic}}

Let $H:\RR_+\to\RR$ be the function defined as
\[H:=\frac{\curlyL_{c_1}f_1\ast\cdots\ast \curlyL_{c_k}f_k}{f_1\ast\cdots\ast f_k}.\]
For each $1\le i\le k$, let $\varphi_i$ be the Laplace transform of $\tau_i$; in particular, $g_i:=\curlyL_{c_i}f_i/\varphi_i(c_i)$ is a p.d.f. supported on $[0,+\infty)$, and
\begin{equation}
\label{ide:nice.magic}
H=\frac{g_1\ast\cdots\ast g_k}{f_1\ast\cdots\ast f_k}\prod_{i=1}^k\varphi_i(c_i).
\end{equation}
Since both the numerator and denominator correspond to the p.d.f.'s of a sum of $k$ non-negative continuous random variables, they are each measurable and almost surely strictly positive and finite. As a result, $H$ is a measurable function.

To complete the proof, it suffices to show that
\begin{equation}
\int_{(t_1,\ldots,t_k)\ge0:\,\sum\limits_{i=1}^kt_i\le a} e^{-\sum\limits_{i=1}^kc_i t_i} \prod_{i=1}^kf_i(t_i)\,dt_i
=\int_0^a H(t)\,(f_1\ast\cdots\ast f_k)(t)\,dt,
\label{ide:show.magic}
\end{equation}
for all $a\ge0$. But this is rather direct because
\[\int_0^a H(t)\,(f_1\ast\cdots\ast f_k)(t)\,dt\qquad\qquad\qquad\qquad\qquad\qquad\qquad\qquad\qquad\qquad\qquad\qquad\qquad\qquad\]
\vspace{-24pt}
\begin{align*}
&=\int_0^a (\curlyL_{c_1}f_1\ast\cdots\ast \curlyL_{c_k}f_k)(t)\,dt \\
&=\int_0^a\,dt\,\int_{(t_1,\ldots,t_{k-1})\ge0:\,\sum\limits_{i=1}^{k-1}t_i\le t} e^{-c_k\big(t-\sum\limits_{i=1}^{k-1}t_i\big)}f_k\!\!\left(t-\sum\limits_{i=1}^{k-1}t_i\right)\prod_{i=1}^{k-1} e^{-c_i t_i} f_i(t_i)\,dt_i \\
&=\int_0^a\,dt\,\int_{(t_1,\ldots,t_{k-1})\ge0:\,\sum\limits_{i=1}^{k-1}t_i\le t} e^{-\sum\limits_{i=1}^{k-1}c_i t_i-c_k\big(t-\sum\limits_{i=1}^{k-1}t_i\big)}f_k\!\!\left(t-\sum\limits_{i=1}^{k-1}t_i\right)\prod_{i=1}^{k-1} f_i(t_i)\,dt_i.
\end{align*}
The identity in (\ref{ide:show.magic}) follows from the Lebesgue-measure-preserving change of variables: $(t_1,\ldots,t_{k-1},t)\longrightarrow(t_1,\ldots,t_{k-1},t_d)$, where $t_d:=t-\sum_{i=1}^{k-1}t_i$, thereby completing the proof.

%%%%%%%%%%%%%%%%%%%%%%%%%%%%%%%
\subsection{Proof of Proposition~\ref{prop:triangle}.}

Before proving the proposition, we state (without proof) a more general form of the memoryless property of the exponential distribution.

\medskip

\begin{lemma}
(General exponential memoryless property.) Let $X,Y,Z$ be random variables such that $(X,Y)$ is independent of $Z$, and   $Z\sim\text{Exponential}(\lambda)$. 
Then, conditioned on having $Z>X$, $Y$ and $(Z-X)$ are independent, with $(Z-X)\sim\text{Exponential}(\lambda)$. In particular, if $X$ and $Y$ are independent, the distribution of $Y$ is unaffected by the conditioning.
\label{lem:memoryless}
\end{lemma}

\medskip

First, observe that
\begin{align*}
\PP(\TT=T_1)
&=\PP\big(\tau_{\{s,o\}}<\tau_{\{s,v\}},\tau_{\{o,v\}}<(\tau_{\{s,v\}}-\tau_{\{s,o\}}\big) \\
&=\PP\big(\tau_{\{s,o\}}<\tau_{\{s,v\}}\big)\cdot\PP\big(\tau_{\{o,v\}}<(\tau_{\{s,v\}}-\tau_{\{s,o\}})\big|\tau_{\{s,o\}}<\tau_{\{s,v\}}\big).
\end{align*}
Due to Lemma~\ref{lem:basic}, the first factor above is $\lambda_1/(\lambda_1+\lambda_2)$, and due to lemmas~\ref{lem:basic}-\ref{lem:memoryless}, the second factor is $\lambda_3/(\lambda_2+\lambda_3)$. Hence:
\begin{equation}
\label{ide:PT=T1}
\PP(\TT=T_1)=\frac{\lambda_1\lambda_3}{(\lambda_1+\lambda_2)(\lambda_2+\lambda_3)}.
\end{equation}
Moreover, given that $\TT=T_1$, the infection time of $o$ is $\tau_{s,o}$ conditioned on the event $\tau_{s,o}<\tau_{s,v}$; in particular, $\tau_o\sim\text{Exponential}(\lambda_1+\lambda_2)$ due to Lemma~\ref{lem:basic}.

Likewise:
\begin{equation}
\label{ide:PT=T2}
\PP(\TT=T_2)=\frac{\lambda_2\lambda_3}{(\lambda_1+\lambda_2)(\lambda_1+\lambda_3)}.
\end{equation}
Moreover, given that $\TT=T_2$, $\tau_o=\tau_{\{s,v\}}+\tau_{\{v,o\}}$ conditioned on having $\tau_{\{s,v\}}<\tau_{\{s,o\}}$ and $\tau_{\{v,o\}}<(\tau_{\{s,o\}}-\tau_{\{s,v\}})$. But, due to the lemmas~\ref{lem:basic}-\ref{lem:memoryless}, when $\tau_{\{s,v\}}<\tau_{\{s,o\}}$, $\tau_{\{s,v\}}\sim\text{Exponential}(\lambda_1+\lambda_2)$ and $(\tau_{\{s,o\}}-\tau_{\{s,v\}})\sim\text{Exponential}(\lambda_1)$ are independent. Hence, again by lemma~\ref{lem:basic}-\ref{lem:memoryless}, when $\tau_{\{s,v\}}<\tau_{\{s,o\}}$ and $\tau_{\{v,o\}}<(\tau_{\{s,o\}}-\tau_{\{s,v\}})$, $\tau_{\{v,o\}}\sim\text{Exponential}(\lambda_1+\lambda_3)$ and it is independent of $\tau_{\{s,v\}}$. Thus, $\tau_o\sim\text{Exponential}(\lambda_1+\lambda_2)\oplus\text{Exponential}(\lambda_1+\lambda_3)$.

From the identities in (\ref{ide:PT=T1})-(\ref{ide:PT=T2}) we obtain that
\begin{equation}
\label{ide:PT=T3}
\PP(\TT=T_3)
=\frac{\lambda_1\lambda_2(\lambda_1 + \lambda_2 + 2\lambda_3)}{(\lambda_1 + \lambda_2)(\lambda_2 + \lambda_3)(\lambda_1 + \lambda_3)}.
\end{equation}
Further, when $\TT=T_3$, $o$ may be infected in two ways. Either $s$ infects $o$ before infecting $v$, in which case $\tau_o\sim\text{Exponential}(\lambda_1+\lambda_2)$; or, $s$ infects $v$ before infecting $o$, in which case $\tau_o\sim\text{Exponential}(\lambda_1+\lambda_2)\oplus\text{Exponential}(\lambda_1+\lambda_3)$. If we define the event $A$ as ``$s$ infects $v$ before infecting $o$'', the conditional probability of $A$ given that $\TT=T_3$ is
\[\PP(A|\TT=T_3)
%=\frac{\PP(A,\TT=T_3)}{\PP(\TT=T_3)} 
=\frac{\frac{\lambda_2}{\lambda_1+\lambda_2}\cdot\frac{\lambda_1}{\lambda_1+\lambda_3}}{\PP(\TT=T_3)}
=\frac{\lambda_2+\lambda_3}{\lambda_1+\lambda_2+2\lambda_3}.
\]
In particular, conditioned on having $\TT=T_3$, $\tau_o$ has the same distribution as $\text{Exponential}(\lambda_1+\lambda_2)\oplus B\cdot\text{Exponential}(\lambda_1+\lambda_3)$, where $B$ is an independent  Bernoulli r.v. with success probability $(\lambda_2+\lambda_3)/(\lambda_1+\lambda_2+2\lambda_3)$, which completes the proof of the proposition.

%\clearpage
%\nocite{*}


\begin{thebibliography}{10}

\bibitem{anderson1992infectious}
Roy~M Anderson and Robert~M May.
\newblock {\em {Infectious diseases of humans: Dynamics and control}}.
\newblock Oxford university press, 1992.

\bibitem{BerKle23}
Caitlin~M. Berry and William Kleiber.
\newblock {Deep variance gamma processes}.
\newblock {\em Stat}, 12(1):e580, 2023.

\bibitem{bertuzzo2008space}
Enrico Bertuzzo, S~Azaele, Amos Maritan, Marino Gatto, I~Rodriguez-Iturbe, and
  Andrea Rinaldo.
\newblock {On the space-time evolution of a cholera epidemic}.
\newblock {\em Water Resources Research}, 44(1), 2008.

\bibitem{bertuzzo2010spatially}
Enrico Bertuzzo, Renato Casagrandi, Marino Gatto, I~Rodriguez-Iturbe, and
  Andrea Rinaldo.
\newblock {On spatially explicit models of cholera epidemics}.
\newblock {\em Journal of the Royal Society Interface}, 7(43):321--333, 2010.

\bibitem{Cor22}
Jem~N. Corcoran.
\newblock {\em {The Simple and Infinite Joy of Mathematical Statistics}}.
\newblock Independently published, 2022.

\bibitem{Cos25}
Devlin Costello.
\newblock {Tree Source Localization}.
\newblock \url{https://github.com/Decos14/tree-source-localization}, 2025.

\bibitem{Dur16}
Richard Durrett.
\newblock {\em {Essentials of Stochastic Processes}}.
\newblock Springer Texts in Statistics. Springer, 3rd edition, 2016.

\bibitem{Elgin2011}
Henry Elgin.
\newblock {\em {Levy Processes and Parameter Estimation by Maximum Empirical
  Likelihood}}.
\newblock PhD thesis, University of Essex, 2011.

\bibitem{eubank2004modelling}
Stephen Eubank, Hasan Guclu, VS~Anil~Kumar, Madhav~V Marathe, Aravind
  Srinivasan, Zoltan Toroczkai, and Nan Wang.
\newblock {Modelling disease outbreaks in realistic urban social networks}.
\newblock {\em Nature}, 429(6988):180--184, 2004.

\bibitem{FeuMcD81}
Andrey Feuerverger and Philip McDunnough.
\newblock {On Some Fourier Methods for Inference}.
\newblock {\em Journal of the American Statistical Association},
  76(374):379--387, 1981.

\bibitem{FeuMcD81b}
Andrey Feuerverger and Philip McDunnough.
\newblock {On the Efficiency of Empirical Characteristic Function Procedures}.
\newblock {\em Journal of the Royal Statistical Society}, 43(1):20--27, 1981.

\bibitem{Jess24}
Julia~M. Jess.
\newblock Source localization in tree infection networks via laplace
  transforms.
\newblock Master's thesis, University of Colorado, 2024.

\bibitem{ji2019properties}
Feng Ji, Wenchang Tang, and Wee~Peng Tay.
\newblock {On the properties of Gromov matrices and their applications in
  network inference}.
\newblock {\em IEEE Transactions on Signal Processing}, 67(10):2624--2638,
  2019.

\bibitem{kermack1927contribution}
William~Ogilvy Kermack and Anderson~G McKendrick.
\newblock {A contribution to the mathematical theory of epidemics}.
\newblock {\em Proceedings of the royal society of london. Series A, Containing
  papers of a mathematical and physical character}, 115(772):700--721, 1927.

\bibitem{Kik47}
Gustav~R. Kirchhoff.
\newblock Ueber die aufl{\"o}sung der gleichungen, auf welche man bei der
  untersuchung der linearen vertheilung galvanischer str{\"o}me gef{\"u}hrt
  wird.
\newblock {\em Annalen der Physik}, 148:497--508, 1847.

\bibitem{LegJou15}
Benjamin Legros and Oualid Jouini.
\newblock {A linear algebraic approach for the computation of sums of Erlang
  random variables}.
\newblock {\em Applied Mathematical Modelling}, 39(16):4971--4977, 2015.

\bibitem{leinhardt2013social}
Samuel Leinhardt.
\newblock {\em {Social networks: A developing paradigm}}.
\newblock Elsevier, 2013.

\bibitem{MadSen87}
D.B. Madan and E.~Seneta.
\newblock {Simulation of Estimates Using the Empirical Characteristic
  Function}.
\newblock {\em International Statistical Review}, 55(2):153--161, 1987.

\bibitem{mollison1977spatial}
Denis Mollison.
\newblock {Spatial contact models for ecological and epidemic spread}.
\newblock {\em Journal of the Royal Statistical Society: Series B
  (Methodological)}, 39(3):283--313, 1977.

\bibitem{newman2002spread}
Mark~EJ Newman.
\newblock {Spread of epidemic disease on networks}.
\newblock {\em Physical review E}, 66(1):016128, 2002.

\bibitem{OConnor22}
Graham~K. O'Connor.
\newblock {\em Tools for Source Localization on Infection Networks}.
\newblock PhD thesis, University of Colorado, 2022.

\bibitem{paluch}
Robert Paluch, {\L}ukasz Gajewski, Krzysztof Suchecki, Boles{\l}aw
  Szyma{\'{n}}ski, and Janusz~A. Ho{\l}yst.
\newblock {Enhancing Maximum Likelihood Estimation of Infection Source
  Localization}.
\newblock In Dariusz Grech and Janusz Mi{\'{s}}kiewicz, editors, {\em
  Simplicity of Complexity in Economic and Social Systems}, pages 21--41, Cham,
  2021. Springer International Publishing.

\bibitem{Pea22}
T.~C. Peachey.
\newblock {A note on the operator norm of the Laplace transformation}.
\newblock {\em Integral Transforms and Special Functions}, 33(9):711--714,
  2022.

\bibitem{PRL}
Pedro~C. Pinto, Patrick Thiran, and Martin Vetterli.
\newblock {Locating the Source of Diffusion in Large-Scale Networks}.
\newblock {\em Phys. Rev. Lett.}, 109:068702, Aug 2012.

\bibitem{Pru18}
Heinz Prüfer.
\newblock {Neuer Beweis eines Satzes über Permutationen}.
\newblock {\em Archiv der Mathematischen Physik}, 27:742--744, 1918.

\bibitem{TRBS}
Zhesi Shen, Shinan Cao, Wen-Xu Wang, Zengru Di, and H.~Eugene Stanley.
\newblock {Locating the source of diffusion in complex networks by
  time-reversal backward spreading}.
\newblock {\em Phys. Rev. E}, 93:032301, Mar 2016.

\bibitem{Van98}
A.~W. van~der Vaart.
\newblock {\em {Asymptotic Statistics}}.
\newblock Cambridge University Press, 1998.

\bibitem{wasserman1994social}
Stanley Wasserman, Katherine Faust, et~al.
\newblock {\em {Social network analysis: Methods and applications}}.
\newblock Cambridge University Press, 1994.

\end{thebibliography}
\end{document}